\documentclass[3p,12pt]{elsarticle}
\usepackage{amsmath}
\usepackage{amssymb}
\usepackage{color}
\usepackage{graphicx}
\usepackage{epsfig}
\usepackage{amsfonts}
\usepackage{float}
\usepackage{epstopdf}
\usepackage{url}
\usepackage{natbib}
\usepackage{mathrsfs}
\usepackage[%
breaklinks%
,colorlinks%
,linkcolor=red
,anchorcolor=blue%
,pagecolor=blue%
,citecolor=blue%
,bookmarks=ture%
]{hyperref}

\newtheorem{theorem}{Theorem}[section]
\newtheorem{corollary}[theorem]{Corollary}
\newtheorem{lemma}[theorem]{Lemma}
\newtheorem{proposition}[theorem]{Proposition}
\newtheorem{definition}[theorem]{Definition}
\newtheorem{remark}[theorem]{Remark}
\newtheorem{example}[theorem]{Example}
\newtheorem{alg}[theorem]{Algorithm}
\numberwithin{equation}{section}
\newproof{proof}{Proof}

\def\de{{\delta}}

\def\F{{\mathcal{F}}}
\def\r{{\mathbf{r}}}
\def\s{{\mathbf{s}}}

\def\p{{\mathbf{p}}}
\def\q{{\mathbf{q}}}
\def\vv{{\mathbf{v}}}
\def\T{{\mathbf{T}}}
\def\al{\boldsymbol{\alpha}}
\def\be{\boldsymbol{\beta}}
\def\ga{\boldsymbol{\gamma}}
\def\ka{\kappa}

\def\PS{{\mathbb{P}}}
\def\PV{{\mathcal V}}
\def\PE{{\mathcal E}}

\begin{document}

\begin{frontmatter}

\title{Certified Approximation of Parametric Space Curves \\with Cubic B-spline Curves}%

\author[SLY]{Liyong Shen}
 \ead{shenly@amss.ac.cn}
 \address[SLY]{School of Mathematical Sciences, Graduate University of
 Chinese Academy of Sciences}

 \author[CMY]{Chun-Ming Yuan}
 \ead{cmyuan@mmrc.iss.ac.cn}
  \author[CMY]{Xiao-Shan Gao}
 \ead{xgao@mmrc.iss.ac.cn}
\address[CMY]{Key Laboratory of Mathematics Mechanization, AMSS,  Chinese Academy of Sciences}

\begin{abstract} Approximating complex curves with simple parametric
curves is widely used in CAGD, CG, and CNC. This paper presents an
algorithm to compute a certified approximation to a given parametric space curve
with cubic B-spline curves. By certified, we mean that the
approximation can approximate the given curve to any given precision
and preserve the geometric features of the given curve such as
the topology, singular points,  etc.
The approximated curve is divided into segments called quasi-cubic
B\'{e}zier curve segments which have properties similar to a cubic
rational B\'{e}zier curve. And the approximate curve is naturally
constructed as the associated cubic rational B\'{e}zier curve of the
control tetrahedron of a quasi-cubic curve. A novel
optimization method is proposed to select proper weights in the
cubic rational B\'{e}zier curve to approximate the given curve. The
error of the approximation is controlled by the size of its
tetrahedron, which converges to zero by subdividing the curve
segments. As an application, approximate implicit equations of the
approximated curves can be computed. Experiments show that the
method can approximate space curves of high degrees with high
precision and very few cubic B\'{e}zier curve segments.
\end{abstract}




\end{frontmatter}

\section{Introduction}
Parametric curves are widely used in different fields such as
computer aided geometric design (CAGD), computer graphics (CG),
computed numerical control (CNC)
systems~\cite{Hoschek1993,Piegl1997}. One basic problem in the study
of parametric curves is to approximate the curve with lower degree
curve segments.
For a given digital curve, there exist  methods to find such
approximate curves efficiently
\cite{Pottmann2002,Renka2005,Aigner2007,koog2009}.
If the curve is given by explicit expressions, either parametric or
implicit, these  methods are still usable. However, some important
geometric features such as singular points cannot be preserved. In
this paper, we will focus on computing approximate curves which can
approximate the given curve to any precision and preserve the
topology and certain geometric features of the given space curve.
Such an approximate curve is called a certified approximation.
Here, the geometric features include cusps, self-intersected points, inflection points, torsion vanishing points, as well as the segmenting points and the left(right) Frenet frames of these points.

There are lots of papers tried to approximate a smooth parametric curve segment~\cite{Hoschek1993,Degen93,Degen05,Farin08,Hijllig95,Xu01,Pelosi05,Rababah2007,Chenxd10}.
Among them, Geometric Hermite Interpolation (GHI) is a typical method for the curve approximation. Degen~\cite{Degen05} presented an overview over the developments of geometric Hermite approximation theory for planar curves. Several 2D interpolation schemes to produce curves close to circles were proposed in~\cite{Farin08}.
The certified approximation were considered by some authors and they focused on the case of planar curves~\cite{Gao2004,Yang2004,Liming2006,Ghosh2007}.

For space curves, Hijllig and Koch~\cite{Hijllig95} improved the standard cubic Hermite interpolation with approximation order five by interpolating a third point. Xu and Shi~\cite{Xu01} considered the GHI for space curves by parametric quartic
B\'{e}zier curve. Pelosi et al.~\cite{Pelosi05} discussed the problem of Hermite interpolation by using PH cubic segments. Chen et al.~\cite{Chenxd10} enhanced the GHI by adding an inner tangent
point and the approximation was then more accurate.
These methods were mainly designed for the local approximation of a parametric curve segment. The approximate curves obtained generally cannot preserve geometric features and
topologies for the global approximation. The algorithms had to be improved to meet certain
special conditions. For instance, Wu et al~\cite{Zhongke2003} presented an
algorithm to preserve the topology of voxelisation and
Chen et al~\cite{Chen2010} gave the formula of the intersection curve of two
ruled surfaces by the bracket method.
As a further development for certified approximation, more properties such as the topology and singularities of the curve need to be discussed in the approximation process. We would like to give the local approximation with certain restrictions. And the local approximation methods can then be used in the global certified approximation naturally.

The certified approximation is also based on the topology determination.
For implicit curves, the problem of topology determination
was studied in some papers such as~\cite{Alcazar2005,chenliang2008,Daouda2008,curve-ib}.
Efficient algorithms were proposed in \cite{wangh2009} and \cite{Rubio2009}
to compute the real singular points of a rational parametric space curve by the $\mu$-basis method
and the generalized $D$-resultant method respectively.
An algorithm was proposed to compute the topology  for a rational parametric space curve \cite{Alcazar2009}.
However, even we have the methods to determine the topology of space curves and the methods to approximate the space curves with free form curves, the combination of them is not straightforward. The topology may change while the line edges in topology graph are replaced by the approximate free form curve segments. For example, some knots may be brought in or lost such that the crossing number of the approximate curve is not equivalent to the approximated curve.

In this paper, we compute a certified approximation to a given parametric space curve
with a rational cubic B-spline curve based on the topology. The cubic rational
B\'{e}zier curve is taken as the approximate curve segment because it is
the simplest non-planar curve and has nice properties~\cite{Forrest1980,chen2002}.
The presented method consists of two major steps.

In the first step, the given space curve segment is divided into
sub-segments which have similar properties to a cubic rational B\'{e}zier curve.
Such curve segments are called quasi-cubic B\'{e}zier curves.
The preliminary work of our division procedure is to compute the singular points
and the topology graph of the given curve, which have already been studied in
\cite{Manocha1992,wangh2009,Rubio2009,Alcazar2009}.
Inflection points and torsion vanishing points of the curve are also
added as character points.
We further divide the curve segments to ensure that the subdivided
curve segments have similar properties to a cubic B\'{e}zier curve.
For instance, each curve segment has an associated control tetrahedron
whose four vertices consist of the two endpoints of the curve segment and the two
intersection points of the tangent lines and the osculating planes at
the different endpoints respectively. And the curve segment is inside its associated
control tetrahedron.
Furthermore, we need to ensure some monotone properties about
the associated control tetrahedron, which
are necessary for the convergence of the algorithm.
The tetrahedrons are then just the control polytope of the approximate cubic B\'{e}zier curves. In other words, the approximate curve is controlled by the sequence of the tetrahedrons. And this property ensure the topological isotopy for the approximated and approximate curves.
Some more careful discussions are proposed for both cubic B\'{e}zier and quasi-cubic curve segments.

In the second step of the algorithm, we use a cubic rational B\'{e}zier spline
to approximate a quasi-cubic B\'{e}zier curve obtained in the first step.
Some different approximation methods can be used here such as GHI with inner tangent points~\cite{Chenxd10}.
However, as we mentioned, a quasi-cubic B\'{e}zier curve has an associated control tetrahedron.
The associated cubic rational B\'{e}zier curve of this tetrahedron is naturally used
as the approximate curve. So, each curve segment and its approximated cubic curve
segment share the same control tetrahedron.
A novel method, called shoulder point approximation,
is proposed to select parameters in the cubic B\'{e}zier curve so that it
can optimally approximate the given curve segment. If the distance between the two
curve segments is larger than the given precision, we further subdivide the given curve segment
and approximate each sub-segment similarly. The error of the approximation is
controlled by the size of the associated tetrahedrons, which are proved to converge to zero.
In the subdivision process, there is one important difference between our algorithm with the others. We only need to check the collision of the sub-tetrahedrons subdivided from which are the intersected before the subdivision, since the sub-tetrahedrons are included in its father tetrahedrons. In general algorithms, one has to check the collision of all pair of the approximate curve segments or their control polytopes after a subdivision.
Finally, the rational cubic B\'{e}zier curves are converted to a $C^1$
rational B-spline with a proper knot selection and used as the final approximate curve.
After a cubic parametric approximate segment is computed, we can compute its algebraic variety using the $\mu$-basis method~\cite{Cox1998}, which can be used as the approximate implicit equations for the
given parametric curve.

The proposed method is implemented and experimental results show that the method
can be used to compute certified approximate curves to high degree space curves efficiently.
The computed rational B-spline has very few pieces and can approximate the given curves with high precision.

The rest of this paper is organized as follows.
In Section 2, some notations and preliminary results are given.
In Section 3, we give the algorithm to compute
the dividing points such that each divided segment is a quasi-cubic curve.
In Section 4, the method of parameter selection for the cubic rational B\'{e}zier segments is proposed and then an
algorithm based on shoulder point approximation is given. We also prove
that the termination of the algorithm. The final
algorithm is given in Section 5, and some examples are used to illustrate the algorithm.
In section 6, the paper is concluded.

\section{Preliminaries}

Basic notations and preliminary results about rational parametric
curves and cubic B\'{e}zier curves are presented in this section.

\subsection{Basic notations}
 A parametric
space curve is defined as
\begin{equation}\label{curve}
\r(t)=(x(t),y(t),z(t)),
\end{equation}
where $x(t),y(t),z(t)\in \mathbb{Q}(t)$ and $\mathbb{Q}$ is the field of
rational numbers.
In the univariate case, L\"{u}roth's theorem provides a proper
reparametrization algorithm and some improved algorithms which can also be
found such as~\cite{Manocha1992}. So we assume that~\eqref{curve} is
a proper parametric curve in an interval $[0,1]$£¬ since any
interval $[a,b]$ can be transformed to $[0,1]$ by a parametric
transformation $t\leftarrow\frac{t-a}{b-a}$. Further, the
denominators of~\eqref{curve} are assumed to have no real roots in
$[0,1]$.

The \emph{tangent vector} of $\r(t)$ is $\r'(t)=(x'(t),y'(t),z'(t))$
and the \emph{tangent line} of $\r(t)$ at a point $\r(t_0)$ is
$\T(t_0)=\r(t_0)+\lambda\r'(t_0),\lambda\in \mathbb{Q}$. A point
$\r(t_0)$ is called \emph{a singular point} if it corresponds to more
than one parameters with multiplicities counted. A singular point is
called a \emph{cusp} if $\r'(t_0)$ is the vector of zeros, which
means that $t_0$ is a multiple parameter; otherwise, it is an
\emph{ordinary singular} point~\cite{Rubio2009}. The curvature and
torsion of the curve are
$$\kappa(t)=\frac{\|\r'(t)\times\r''(t)\|}{\|\r'(t)\|^3},\
\tau(t)=\frac{(\r',\r'',\r''')}{\|\r'\times\r''\|}.$$
A point is called an \emph{inflection} if its curvature is zero and
called \emph{torsion vanishing point} if its torsion is zero. All
these points are called \emph{character points} of the curve, and
$\r(t)$ is a \emph{normal curve} if it has a finite number of
character points. A rational space curve is always a normal curve. In this
paper, we assume that $\ka(t)\not\equiv 0$ and $\tau(t)\not\equiv
0$, which means that the curve is not a planar curve.

If $\r(t_0)$ is not a character point, then the Frenet frame at
$\r(t_0)$ can be defined as
$\F(t_0):=\{\r(t_0);\al(t_0),\be(t_0),\ga(t_0)\}$ where
$\al(t_0)=\frac{\r'(t_0)}{\|\r'(t_0)\|}$,
$\be(t_0)=\ga(t_0)\times\al(t_0)$,
$\ga(t_0)=\frac{\r'(t_0)\times\r''(t_0)}{\|\r'(t_0)\times\r''(t_0)\|}$
are the unit tangent vector, unit principal normal vector, and unit
bi-normal vector, respectively. And the osculating plane is
$O(t_0):=((x,y,z)-\r(t_0))\cdot \ga(t_0)=0$.

For a point with $\ka(t_0)=0$, the bi-normal vector is not defined,
neither is the osculating plane. Here, we define them using limit.
Consider the limit $\lim_{t\to t_0} \ga(t)$ of the bi-normal vector
at~$t_0$. Since the left limit and the right limit are generally
different, we define the {\em left bi-normal vector} and the {\em
right bi-normal vector} as $\ga^{-}(t_0):=\lim_{t\to t_0-0} \ga(t)$
and $\ga^{+}(t_0):=\lim_{t\to t_0+0} \ga(t)$ respectively. The
limitations always exist if $\r(t)$ is a rational space curve of
form \eqref{curve}.
As a consequence, the left and right osculating planes at~$t_0$ are
$O^{-}(t_0):=((x,y,z)-\r(t_0))\cdot \ga^{-}=0$ and
$O^{+}(t_0):=((x,y,z)-\r(t_0))\cdot \ga^{+}=0.$ If the $\ka(t_0)\neq
0$, one can find that $\ga^{+}(t_0)=\ga^{-}(t_0)$ and
$O^{+}(t_0)=O^{-}(t_0)$.

Similarly, if $t_0$ is at a cusp, we define the left and right
tangent vectors as $\al^{-}(t_0):=\lim_{t\to t_0-0}\al(t)$ and
$\al^{+}(t_0):=\lim_{t\to t_0+0}\al(t)$, respectively. Hence, the
corresponding left and right principal vectors are
$\be^{-}(t_0):=\ga^{-}(t_0)\times\al^{-}(t_0)$ and
$\be^{+}(t_0):=\ga^{+}(t_0)\times\al^{+}(t_0)$. We also denote the
left and right tangent lines as
$\T^{-}(t_0)=\r(t_0)+\lambda\al^{-}(t_0)$ and
$\T^{+}(t_0)=\r(t_0)+\lambda\al^{+}(t_0)$ where $\lambda$ is the
real number parameter. Then, a rational parametric curve $\r(t)$
always has left and right Frenet frames.

\subsection{Rational cubic B\'{e}zier curve }
A rational B\'{e}zier curve with degree $n$ has the following form
$$\p(t)=\frac{\sum_{i=0}^n\omega_i\p_iB_i^n(t)}{\sum_{i=0}^n\omega_iB_i^n(t)},\ t\in [0,1],
$$
where $\omega_i\ge 0$ are associated weights of the control points
$\p_i\in \mathbb{R}^3$ and $B_i^n(t)=\binom{n}{i}(1-t)^{n-i}t^i$.
When $n=3$, it defines a cubic rational B\'{e}zier curve where
$\lozenge\p_0\p_1\p_2\p_3$ is called the control tetrahedron of
$\p(t)$. One can set the weight $\omega_0=\omega_3=1$ up to a
parametric transformation. We now consider the cubic curve and omit
superscript $3$ from $B_i^3(t)$
\begin{equation}\label{Bezier}
\p(t)=\frac{\p_0B_0(t)+\omega_1\p_1B_1(t)+\omega_2\p_2B_2(t)
+\p_3B_3(t)}{B_0(t)+\omega_1B_1(t)+\omega_2B_2(t)+B_3(t)},\
t\in [0,1].
\end{equation}

The rational cubic B\'{e}zier curve~\eqref{Bezier} has the following
properties.

\begin{lemma}\label{Bezier-lm} Let $\p(t)$ be a non-planar cubic rational curve of
the form~\eqref{Bezier}. Then
\begin{description}
   \item[1)] $\p(t)$ passes through the endpoints $\p_0,\p_3$ with the
  corresponding tangent directions $\p'(0)$ and $\p'(1)$ parallel to
  $\p_0\p_1$ and $\p_2\p_3$ respectively.
  \item[2)] $\p_0\p_1\p_2$ and $\p_1\p_2\p_3$ are
  the osculating planes of $\p(t)$ at the endpoints $\p_0$ and $\p_3$, respectively.
  \item[3)] $\p(t)$ lies inside its control tetrahedron
  $\lozenge\p_0\p_1\p_2\p_3$.
  \item[4)] $\p(t)$ has no singular points and $\ka(t)\neq 0,
  \tau(t)\neq0$ in $[0,1]$.
  \item[5)] For any $t^{\star}_1<t^{\star}_2\in[0,1]$, the control tetrahedron of
  $\p^{\star}(t)=\p(t),t\in[t^{\star}_1,t^{\star}_2]$ is inside the control
  tetrahedron of $\p(t)$ .
  \item[6)] $\|\p_0\p_{01}\|, \|\p_1\p_{12}\|$, and $\|\p_2\p_{23}\|$ are strictly
  monotone for $t^{\star}\in (0,1)$ where $\p_{01},\p_{12}$, and $\p_{23}$ are the
  intersection points of the osculating plane $O(t^{\star})$ with $\p_0\p_1,\p_1\p_2,$
  and $\p_2\p_3$ respectively.
   \item[7)]  $\|\p_0\p_{03}\|$ and $\|\p_1\p_{12}\|$ are strictly monotone
   for
   $t^{\star}\in (0,1)$ where $\p_{03}=\p_1\p_2\p(t^{\star})\bigcap
   \p_0\p_3$ and $\p_{12}=\p_0\p_3\p(t^{\star})\bigcap \p_1\p_2$.
  \end{description}
\end{lemma}

\begin{proof}
Properties 1), 2) and 3) are basic properties of B\'{e}zier curves and
the proof can be founded in~\cite{Hoschek1993}. They also can be
checked directly.


For 4), Li and Cripps shown that there is no cusps and inflection
points for a non-degenerate rational cubic space curves
in~\cite{Li1997}, and the torsion can be checked directly. Wang et
al. also proved that a cubic space curve has no singular points by
moving planes method in~\cite{wangh2009}.

5) can be proved by a successive Decasteljau
subdivision~\cite{Hoschek1993}. The control tetrahedron of
$\p^{\star}_1(t),t\in[t_1^{\star},1]$ is inside the control
tetrahedron of $\p(t)$. Successively, the control tetrahedron of
$\p^{\star}(t),t\in[t_1^{\star},t_2^{\star}]$ lies in the control
tetrahedron of $\p^{\star}_1(t)$.

Property 6) can be derived from the above five properties. Also this
property is a special case of the following Theorem~\ref{sub-thm} in
this paper.

For 7), it is sufficient to prove that the planes
$\p_1\p_2\p(t^{\star})$ and $\p_0\p_3\p(t^{\star})$ do not touch
$\p(t^{\star})$ with  $t^{\star}\in (0,1)$, respectively. Since
$\p_0\p_3\p(t^{\star})$ passes through $\p_0,\p_3$ and $\p(t)$ is
cubic, $\p_0\p_3\p(t^{\star})$ cannot have any tangent point
different from $\p_0,\p_3$. Supposing the plane
$\p_1\p_2\p(t^{\star})$ touches $\p(t^{\star})$ at  $t^{\star}\in
(0,1)$, the osculating plane $O(t^{\star})$ must intersects
$\p_1\p_2\p(t^{\star})$ with the tangent line $\T(t^{\star})$. By
6), $\T(t^{\star})$ must intersect $\p_1\p_2$ which is the
intersection line of $O(0), O(1)$. However, according to Decasteljau subdivision, the intersection point of $\T(t^{\star})$ and $O(0)$ is always different from that of $\T(t^{\star})$ and $O(1)$. Then there is a contradiction.
\hfill\qed\end{proof}

The shoulder point of a cubic B\'{e}zier curve will play an important
role~\cite{Forrest1980}. The definition is given below.

\begin{definition} Let
 $\p(t)$ be a curve of the form~\eqref{Bezier}. Its shoulder
 point $\s$ is defined as intersection point of $\p(t)$ and the plane $\p_1\p_2\p_M$ where
 $\p_M=(\p_0+\p_3)/2$ (Figure~\ref{fig1}).
\end{definition}
\begin{figure}[!h]
 \centering
 \includegraphics[width=0.5\textwidth]{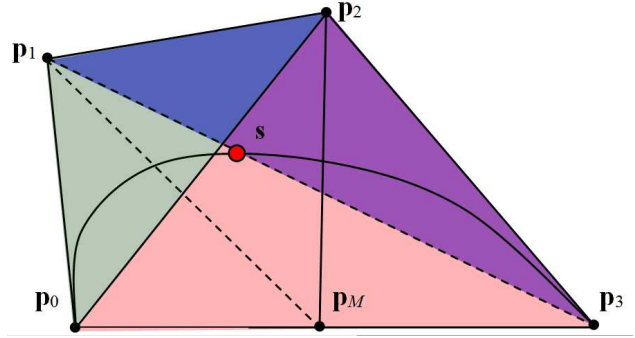}
 \caption{ \label{fig1} Shoulder point of a B\'{e}zier cubic curve}

\end{figure}
\begin{proposition}\label{shoulder} Let $\s$ be the shoulder point of
$\p(t)$. Then
$\s=\p(1/2)=\lambda_1\p_1+\lambda_2\p_2+(1-\lambda_1-\lambda_2)\p_M$
where $\lambda_1={\frac {3\omega_1}{2+3\,\omega_1+3\,\omega_2}}$,
$\lambda_2={\frac {3\omega_2}{2+3\,\omega_1+3\,\omega_2}}$.
\end{proposition}
\begin{proof}
By 7) of Lemma~\ref{Bezier-lm}, there exists a unique intersection
point of $\p(t)$ and the plane $\p_1\p_2\p_M$. And
$\lambda_1,\lambda_2$ and $1-\lambda_1-\lambda_2$ are just the area
coordinates of $\s$ in the triangle $\p_1\p_2\p_M$. More details can
be found in~\cite{Forrest1980}. \hfill\qed\end{proof}

It is known that the curve is closer to the control point when its
associated weight is greater. We now consider the point which has
the maximum distance to the planes $P_1=\p_0\p_2\p_3$ and
$P_2=\p_0\p_1\p_3$ respectively.

\begin{definition}
Let $\r(t),t\in[0,1]$ be a curve segment on the same side of a plane
$Q$ with the two endpoints on $Q$. For another plane $R$ parallel
to $Q$, a tangent point of $\r(t)$ with the plane $R$  is called a {\em
parallel point} of $\r(t)$ associated to the plane $Q$.
\end{definition}

According to the definition, a parallel point should satisfy
\begin{equation}
  \label{parallel}
  |\r'(t),\q_1-\q_0,\q_2-\q_0|=0,
  \end{equation}
  where $\q_0,\q_1$ and $\q_2$ are three non co-linear points on $Q$.
In general, there may be several parallel points for a curve
segment and a fixed plane. However, for the rational cubic curve
segment~\eqref{Bezier}, there is a unique parallel point associated
to $P_1=\p_0\p_2\p_3$, and similarly, there is a unique
parallel point associated to $P_2=\p_0\p_1\p_3$.

 \begin{proposition}\label{oneparall}
Let $\p(t)$ be a non-planar cubic rational curve of the
form~\eqref{Bezier}. Then there are unique parallel points
associated to the planes $P_1$ and $P_2$ respectively, and they are
points of $\p(t)$ having the maximal distance to $P_1$ and $P_2$
respectively.
 \end{proposition}
 \begin{proof}
By equation~\eqref{parallel}, we can find that $\frac{3t^3-6t^2+6t
-2}{3t(t-1)^2}=\omega_1$ and
$\frac{3t^3-3t^2+3t-1}{3t^2(t-1)}=\omega_2$ are the constraint
equations for the parallel points associated to $P_1$ and $P_2$
respectively. They are two monotone functions for $t\in(0,1)$ with
two asymptotes $t=0, 1$. It means that for any weights there is only
one parallel point associated to $P_i$. Furthermore, the parallel
point has the maximal distance since the endpoints of the curve are
on $P_i$.
 \hfill\qed\end{proof}

\section{Quasi-cubic segments on space parametric curves}
In this section, we propose a method to divide a given curve $\r(t)$
into segments which have similar properties to cubic B\'{e}zier curves,
which are called quasi-cubic B\'{e}zier segments and can be approximated
by cubic rational B\'{e}zier curves nicely.

\subsection{Conditions for subdivision}

Let $t_0,t_1$ be the endpoints of a curve segment $\r(t)$. We will
define an associated tetrahedron for it. Let $O^{+}(t_0)$ and
$O^{-}(t_1)$ be the right and left osculating planes at the
endpoints respectively. We denote their intersection line as $L$, if
they are not parallel. Since $L$ and the right tangent line
$\T^{+}(t_0)$ are coplanar, they intersect at a point $\r_1$ if they
are not parallel. Similarly,
 $L$ and the right tangent line $\T^{-}(t_1)$ intersect at a point
 $\r_2$ if they are not parallel.
So we obtain an \emph{associated tetrahedron}
$\lozenge(t_0,t_1)=\lozenge\r_0\r_1\r_2\r_3$  where
$\r_0=\r(t_0)$ and $\r_3=\r(t_1)$ if $\r_1\neq \r_2$.

We have shown that a cubic B\'{e}zier curve segment has eight properties
in Lemma~\ref{Bezier-lm} and Proposition~\ref{oneparall}. In the
following, we will show how to divide any given rational curve
segment into sub-segments  having similar properties.
\begin{definition}
A curve segment is called a {\em quasi-cubic B\'{e}zier curve segment},
or simply a \emph{quasi-cubic segment}, if it has the eight properties in
Lemma~\ref{Bezier-lm} and Proposition~\ref{oneparall}.
\end{definition}

\begin{theorem}\label{quasi-cubic}
  Given $\r(t)$ and $t_0$, there always exists $t_1>t_0$ such that $\r(t),t\in[t_0,t_1]$ is a quasi-cubic B\'{e}zier curve segment.
\end{theorem}

We leave the proof of this theorem at the end of the subsection~\ref{subsec-3.2}.

\begin{definition}Let $\r(t),t\in[t_0,t_1]$ be a quasi-cubic segment. Then its \emph{associated cubic B\'{e}zier curve segment} is defined by the associated tetrahedron of $\r(t)$, i.e., the control points are $\r_0,\r_1,\r_2$ and $\r_3$.
\end{definition}

In order to divide the curve segment into quasi-cubic segments, we
first add the inflection points and torsion vanishing points as the
dividing points, denoted by $\PS$. The parameters of these points
can be computed by solving the real roots of $\ka(t)\tau(t)=0$.
The left and right Frenet frames are also needed.  There are several
efficient methods to find the real roots of a univariate
polynomial~\cite{Rouillier2004,cheng2009} and one can use the
procedures \texttt{realroot} and \texttt{isolate} in \texttt{Maple}.

We need to find more dividing points. Fix a start point $t=t_0$, we
now try to determine $t_1$ such that $t_1-t_0$ is as big as possible
and the segment is included in its associated tetrahedron designed
above.
Several boundary parametric values to exclude some special points
with respect to $t_0$ are computed in the following cases:

\textbf{Condition I}). Let $t^\star_1>t_0$ be its nearest parametric
value from $\PS$. Find $t_1\in (t_0,t_1^{\star})$ such that
$F_1(s_1,s_2):=\al^{+}(s_1)\cdot \ga^{-}(s_2)\neq 0$ and
$F_2(s_1,s_2):=\al^{-}(s_2)\cdot \ga^{+}(s_1)\neq 0$
for any $t_0 \leq s_1< s_2\leq t_1$, meaning that the right tangent
vector $\al^{+}(s_1)$ is not parallel to the left osculating plane
$O^{-}(s_2)$ and the left tangent vector $\al^{-}(s_2)$ is not
parallel to the left osculating plane $O^{+}(s_1)$.

Since the curve is non-planar, $F_i(s_1,s_2),i=1,2$ cannot be
identically zero. We take a further look at the inequalities $F_1\ne
0, F_2\ne0$. Since the derivative can be computed using limits,
$\r(t)$ is differentiable to any order although the left and right
derivative may be different. For conveniences, we omit the $+,-$
marks to distinguish between left and right derivatives. In what
below, we give detailed analysis for $F_1$ and the analysis of $F_2$
is similar.
$$F_1(s_1,s_2)=\al(s_1)\cdot
\ga(s_2)=\frac{|\r'(s_1),\r'(s_2),\r''(s_2)|}{\|\r'(s_1)\|\|\r'(s_2)\times\r''(s_2)\|}.$$
Assuming $s_1=t_0+\de_1,s_2=s_1+\de_2,\de_1\ge 0,\de_2>0$,
$F_1(s_1,s_2)$ is re-parameterized as
$$ F_1(\de_1,\de_2) =\frac{|\r'(t_0+\de_1),\r'(t_0+\de_1+\de_2),\r''(t_0+\de_1+\de_2)|}
{\|\r'(t_0+\de_1)\|\|\r'(t_0+\de_1+\de_2)\times\r''(t_0+\de_1+\de_2)\|}.$$
Expanding the vectors of the numerator at $t=t_0+\de_1$ as Taylor
series $\r'_S(t_0+\de_1),\r'_S(t_0+\de_1+\de_2)$ and
$\r''_S(t_0+\de_1+\de_2)$ respectively, and combining them, we have
\begin{equation}\label{taylor_exp}
 F_1(\de_1,\de_2)=\frac{\de_2^2|\r'_S(t_0+\de_1),\tilde\r''_S(t_0+\de_1+\de_2),
\tilde\r'''_S(t_0+\de_1+\de_2)|}
{\|\r'(t_0+\de_1)\|\|\r'(t_0+\de_1+\de_2)\times\r''(t_0+\de_1+\de_2)\|},
\end{equation}
where
$\tilde\r''_S(t_0+\de_1+\de_2)=(\r'_S(t_0+\de_1+\de_2)-\r'_S(t_0+\de_1))/\de_2$
and
$\tilde\r'''_S(t_0+\de_1+\de_2)=(\r''(t_0+\de_1+\de_2)-\tilde\r''_S(t_0+\de_1+\de_2))/\de_2$.
Furthermore, when $\de_2=0$,
$\tilde\r''_S(t_0+\de_1)=\r''_S(t_0+\de_1)$ and
$\tilde\r'''_S(t_0+\de_1)=\r'''_S(t_0+\de_1)$.

Let  $f_1(\de_1,\de_2)=F_1(\de_1,\de_2)/\de_2^2$. Then
$f_1(\de_1,0)=\tau(t_0+\de_1)/\|\r'(t_0+\de_1)\|$.
 $ F_1(\de_1,\de_2)=0$ is a planar curve in the plane of $(\de_1,\de_2)$ which has two components:
a double line $\de_2^2=0$ and another planar curve
$f_1(\de_1,\de_2)=0$. That means $f_1=0 $ intersects $\de_2=0$ with
the points which are exactly the torsion vanishing points
$\tau(t_0+\de_1)=0$ of $\r(t)$. And we need not compute these
points since they are already included in the separating points
needed in the topology computation which is discussed in Section
\ref{sec-sa3}.
Consider the intersection points of $f_1(\de_1,\de_2)$ and
$\de_1=0$. We can find that the real roots of $f_1(0,\de_2)=0$ are
associated to the
  vector $\al(s_1) = \r'(t_0)$
 just parallelling to the osculating plane $O(s_2) = O(t_0+\de_2)$.

Thus, condition I) can be reduced to solve the following
optimization problem
\begin{equation}\label{bound-opt}
  \begin{array}{cl}
  \min & \de_1+\de_2\\
  \mbox{s.t.} &  F_1(\de_1,\de_2)=0, \de_1\ge0, \de_2>0
\end{array}
\end{equation}
and then $t_1$ can be selected from $(t_0,t_0+\de_1+\de_2)$. There
are numerical methods to solve the optimization problem. However, we
prefer to solve it based on the above discussion since it is enough
to get a boundary parametric value less than the exact solution
of~\eqref{bound-opt}. We can find the positive real roots of
$f_1(\de_1,0)$ and $f_1(0,\de_2)$ for $\de_1$ and $\de_2$
respectively. Let $\de_1^{\star}$ be the minimal one among all the
real roots. Then $\de_1+\de_2=\de_1^{\star}$ defines a line. If the
line does not intersect $f_1$ in the first quadrant, then $t_1$ can
be in $(t_0,t_0+\de_1^{\star})$. This can be checked by finding the
real roots of $f_1(\de_1^{\star}-\de_2,\de_2)=0$. Otherwise, set
$\de_1^{\star}\leftarrow \de_1^{\star}/2$ and check the process
repeatedly until the proper $\de_1^{\star}$ is found. If
$f_1(\de_1,0)$ and $f_1(0,\de_2)$ have no positive real roots,
$\de_1^{\star}$ can be initialed as $\de_1^{\star}=t_1^\star-t_0$.

Similarly, we can find such a $\de_2^{\star}$ for $F_2$. Finally,
let $t_2^{\star}=\min(t_0+\de_1^{\star}, t_0+\de_2^{\star})$ be the
boundary parametric value of~$t_1$.

\begin{remark}
The function $F_1(\de_1,\de_2)$ in~\eqref{taylor_exp} actually has a
finite number of terms if the approximated curve $\r$ is a rational
curve. If $\r$ is a parametric curve in elementary functions,
$F_1(\de_1,\de_2)$ will be in the series form. However,
 the problem~\eqref{bound-opt} can still be solved
using a numerical method. Starting with an initial value
$\de_1^{\star}$, we can find a boundary number by checking whether
$\de_1+\de_2=\de_1^{\star}$ and $F_1(\de_1,\de_2)$ have common
points in the first quadrant with one of the directions
$\{\de^0\leftarrow \de_1^{\star}/2,\de^0\leftarrow
2\de_1^{\star}\}$.
\end{remark}

Further restrictions will be  proposed afterward. We will omit the
similar discussions and solving processes and give the conditions
directly.

\textbf{Condition II}). Let $t_2^*$ be the parametric value $t_1$
computed in the above procedure. Find $t_1\in (t_0,t_2^{\star})$
such that
$$F(s_1,s_2):=\al^{+}(s_1)\times(\r(s_2)-\r(s_1))\cdot\al^{-}(s_2)\neq 0$$
for any $t_0 \leq s_1< s_2\leq t_1$, which means that the right
tangent line $\T^{+}(s_1)$ and the left tangent line $\T^{-}(s_2)$
are not coplanar.

\textbf{Condition III}). Let $t_3^*$ be the parametric value $t_1$
computed in the above procedure. We should find $t_1\in
(t_0,t_3^{\star})$ such that $F_1(s_1,s_2):=O^{-}(s_2)(\r(s_1))\neq
0$ and $F_2(s_1,s_2):=O^{+}(s_1)(\r(s_2))\neq 0$, which imply that
$\r(s_1)$ is not on the left osculating plane $O^{-}(s_2)$ and
$\r(s_2)$ is not on the right osculating plane $O^{+}(s_1)$.

Conditions I), II), and III) are used to guarantee that the
tetrahedron $\lozenge\r_0\r_1\r_2\r_3$ is not degenerated to a plane
polygon. However, these conditions are still not sufficient for the
curve segment lying inside $\lozenge\r_0\r_1\r_2\r_3$. We will give
one more condition such that the curve segment lies inside the
tetrahedron and has only one parallel points associated to planes
$P_1$ and $P_2$ respectively.

Let $\tilde t_1 <t^{\star}_4$ where  $t^{\star}_4$ is the parameter
value obtained from III). Then the curve segment $\r(t),
t\in[t_0,\tilde t_1]$ satisfies the conditions of I) to III) and
$\r(t)$ has no character points. We will try to find
$t^{\star}\in(t_0,\tilde t_1]$ such that for any
$s_1<s_2<s_3\in[t_0,t^{\star}]$, the tangent vectors $\al(s_1)$,
$\al(s_2)$, and $\al(s_3)$ are not coplanar, i.e.,
\begin{equation}\label{neq}
|\al(s_1), \al(s_2),\al(s_3)|\neq 0.
\end{equation}

The following lemma is needed for further discussion.
\begin{lemma}\label{nosolution}
For a fixed $t_0$ and $\forall \epsilon>0$,
$F(s_1,s_2):=|\al(t_0),\al(s_1), \al(s_2)|=0$ has solutions
$(s_1,s_2)$ in $(0,\epsilon)^2$ if and only if $\r(t)$ is a planar
curve.

\end{lemma}
\begin{proof}
It can be checked by expanding vectors to Taylor series which are
partly illustrated above. \hfill\qed\end{proof} And the lemma also
holds for $F$ mentioned in I) to III). It means that $F(s_1,s_2)$
has no branch segment on the first quadrant of the $(s_1,s_2)$ plane
connecting the origin point.

\textbf{Condition IV}). Find $t^{\star}\in (t_0,\tilde t_1)$ such
that $F:=|\al(s_1), \al(s_2),\al(s_3)|\neq 0$ for any
$s_1<s_2<s_3\in[t_0,t^{\star}]\subset[t_0,\tilde t_1]$. That means
$\r(t)$ does not have a triple of linear dependent tangents in
$[t_0,t^{\star}]$. Suppose $s_1=t_0+\de_1 $, $s_2=s_1+\de_2 $ and
$s_3=s_2+\de_3$ where $\de_1\ge 0, \de_2>0$ and $\de_3>0$.

If $\de_1>0$, then we need to find the least $t_0+\de_1+\de_2+\de_3$
with $F(\de_1,\de_2,\de_3)=0$, that is,
$$\begin{array}{rl}
  \min &  \de_1+\de_2+\de_3\\
 \mbox{ s.t.}&F(\de_1,\de_2,\de_3)=0, \de_1,\de_2,\de_3>0.
\end{array}$$
By Taylor expansion, we find that $F(\de_1,\de_2,\de_3)$ has no
branch passing through the $(\de_i,\de_j)$ plane from the first
octant in the space of $(\de_1,\de_2,\de_3)$. Then we initialize
$\de_i,i=1,2,3$ in the plane $\de_1+\de_2+\de_3=\de_1^{\star}=\tilde
t_1$ and check the intersection of the plane with $F$. Set the
boundary parametric value $t_{51}^{\star}=\de_1^{\star}$ if there is
no intersection; otherwise set $\de_1^{\star}\leftarrow
\de_1^{\star}/2$ and repeat the checking process.

If $\de_1=0$, then $F(\de_2,\de_3)$ degenerates to the special case
mentioned in Lemma~\ref{nosolution} and we can find a boundary
parametric value as $t_{52}^{\star}$. Finally, let
$t^{\star}=\min(t_{51}^{\star},t_{52}^{\star})$.

We  have the following key theorem.
\begin{theorem}
\label{divide-thm}Let $t^\star$ be found by the above process. For
any $\epsilon>0$, $t_1=t^{\star}-\epsilon>t_0$, the associated
tetrahedron $\lozenge\r_0\r_1\r_2\r_3$ of $\r(t),t\in [t_0,t_1]$ is
not degenerated. Furthermore,

\begin{description}
  \item[1)]  $\r(t)$ passes through the endpoints $\r_0,\r_3$ with the
  corresponding tangent directions $\r'(t_0)$ and $\r'(t_1)$ parallel to
  $\r_0\r_1$ and $\r_2\r_3$ respectively.
  \item[2)]  $\r_0\r_1\r_2$ and $\r_1\r_2\r_3$ are
  the osculating planes of $\r(t)$ at the endpoints $\r_0$ and $\r_3$, respectively.
  \item[3)]  $\r(t)$ lies inside its control tetrahedron
  $\lozenge\r_0\r_1\r_2\r_3$.
  \item[4)]  $\r(t)$ has no singular points and $\ka(t)\neq 0,
  \tau(t)\neq0$ in $[t_0,t_1]$.
  \item[5)] There exists only one parallel point between $\r_1$ and $\r_0\r_2\r_3$,
  same to $\r_2$ and $\r_0\r_1\r_3$.
  \end{description}
\end{theorem}
\begin{proof}
According to conditions I) to III), the tetrahedron
$\lozenge\r_0\r_1\r_2\r_3$ does not degenerate. 1), 2), and 4) are
also followed by the discussions.

The curve segment is inside the tetrahedron. We claim that the curve
segment and $\r_3$ are on the same side of plane $P_3=\r_0\r_1\r_2$.
Otherwise, there exists a parallel point $\p$ associated to $P_3$
but on the different side with $\r_3$, since $\r(t)$ is a smooth
segment. Then $\al(\p)$ is parallel to $P_3$ which contradicts to
I). Similarly, the curve and $\r_0$ are on the same side of
$P_0=\r_1\r_2\r_3$.
Furthermore, the curve and $\r_1$ are on the same side of
$P_1=\r_0\r_2\r_3$. Otherwise, there exist at least two parallel
points $\p_1,\p_2$ on different sides of $P_1$. Then
$|\al(\p_1),\al(\p_2),\al(\r_3)|=0$ which contradicts to condition
IV). Similarly, the curve and $\r_2$ are on the same side of
$P_2=\r_0\r_1\r_3$. Therefore, 3) is followed.

Finally, 5) is correct. Otherwise, there exist at least two parallel
points associated to $P_1$ or $P_2$ which will lead a contradiction
to condition IV). \hfill\qed\end{proof}

\begin{proposition}\label{sub-prop} For any $t^{\star}_1<t^{\star}_2\in[t_0,t_1]$,
the sub-tetrahedron
$\lozenge\r_0^{\star}\r_1^{\star}\r_2^{\star}\r_3^{\star}$ of the
sub-segment $\r^{\star}(t),t\in[t^{\star}_1,t^{\star}_2]$ also has
the properties listed in Theorem~\ref{divide-thm}.
  \end{proposition}
\begin{proof}
In the dividing process, the conditions in I) to IV) are satisfied
for the parameters through the interval not just only for the
endpoints. Then the properties are all satisfied within
 $[t_1^{\star},t_2^{\star}]\subset[t_0,t_1]$.
\hfill\qed\end{proof}

\subsection{Further properties of the divided segment}\label{subsec-3.2}
In this subsection, we prove that the curve segment obtained in the
preceding section also has properties 6) and 7) in Lemma
\ref{Bezier-lm}. Before that, we need some preparations.

Suppose that the curve segment $\r(t),t\in[t_0,t_1]$ satisfies
conditions I) - IV) in the preceding section.

\begin{lemma}\label{lemma-pre1}
Let $\lozenge\r_0\r_1\r_2\r_3$ be the control tetrahedron of a given
curve segment $\r(t),t\in[t_0,t_1]$. Then for any $t^\star
\in(t_0,t_1)$, the control tetrahedron
$\lozenge\r_0\r_1^\star\r_2^\star\r_3^\star$ of the curve segment
$\r(t),t\in[t_0,t^\star]$ has the following properties:
\begin{enumerate}
  \item $\r_1^\star$ and $\r_1$ are  on the same side of $\r_0$ in
  the tangent line $\T(t_0)$;
    \item $\r_2^\star$ and $\r_2$ are on the same side of $\T(t_0)$ in
    the osculating plane $O(t_0)$.
  \end{enumerate}
\end{lemma}
\begin{proof} Using the first and second order Taylor expansion of $\r(t)$, one
can prove the lemma. \hfill\qed\end{proof}

\begin{lemma}\label{lemma-pre2}
Let $O(t^\star)$ be the osculating plane of curve $\r(t)$ at
$t^\star\in[t_0,t_1]$. If $\r(t)$ does not pass through
$O(t^\star)$, then $\tau(t^\star)=0$.
\end{lemma}
\begin{proof} Similar to the discussions of condition I), using the third order Taylor expansion,
one can see that $|\r'(t^\star),\r''(t^\star),\r'''(t^\star)|=0$,
that is $\tau(t^\star)=0$. \hfill\qed\end{proof}

We now prove another key property for the  curve segments.

\begin{theorem}\label{sub-thm} Let
$\lozenge\r_0\r_1\r_2\r_3$ be the associated tetrahedron of a curve
segment $\r(t),t\in[t_0,t_1]$. Then $\|\r_0\r_{01}\|,
\|\r_1\r_{12}\|$, and $\|\r_2\r_{23}\|$ are strictly monotone in
$(t_0,t_1)$ where $\r_{01},\r_{12}$, and $\r_{23}$ are the
intersection points of the osculating plane $O(t^{\star})$ and
$\r_0\r_1,\r_1\r_2$, and $\r_2\r_3$ respectively.
\end{theorem}
\begin{proof}

Firstly, the intersection point $\r_{01}$ of $\r_0\r_1$ and the
osculating plane $O(t^{\star})$ must be on the same side with $\r_1$
with respect to $\r_0$ on the curve segment. Otherwise, subdividing
$\r(t)$ at $t^{\star}$,  the sub-segment $\r_1^{\star}(t),
t\in[t_0,t^{\star}]$ will not be inside its tetrahedron for
$\r_{01}\neq \r_0$ by Lemma \ref{lemma-pre1}. We denote by $\r_{02}$
the intersection point of line $\r_0\r_2$ and $O(t^\star)$.
Similarly, $\r_{23}$ is on the same side with $\r_2$ with respect to
$\r_3$ and $\r_{02}$ is on the same side with $\r_2$ w.r.t. $\r_0$
(See Figure~\ref{figproof}).
  \begin{figure}[!h]
 \centering
 \includegraphics[width=0.50\textwidth]{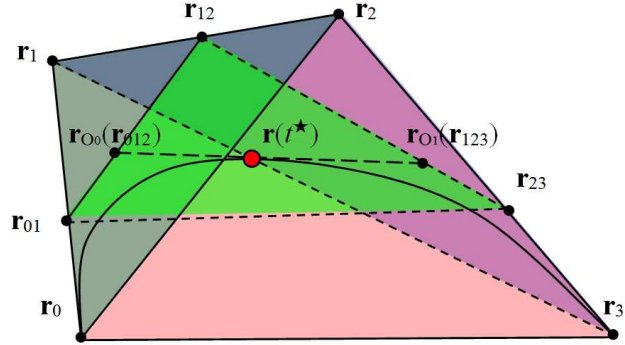}
 \caption{The osculating plane}
 \label{figproof}
\end{figure}

Secondly, we claim that there exist no $t_1^{\star}<t_2^{\star}$ in
$[t_0,t_1]$ such that the osculating planes $O(t_1^{\star})$ and
$O(t_2^{\star})$ have the same intersection point $\r_{01}$ with
$\r_0\r_1$. It is sufficient to prove that there has no
$t^{\star}\in(t_0,t_1)$ such that the osculating plane
$O(t^{\star})$ passes through $\r_1$ by assuming $t_2^{\star}=t_1$
and denote $t_1^{\star}$ by $t^{\star}$. Otherwise, if the
osculating plane $O(t^{\star})$ passes through $\r_1$, then
$O(t^{\star})$ passes through the line $\r_1\r(t^{\star})$ but cannot pass through $\r_0$ and $\r_3$ by the restrictions in condition
I).
Hence $O(t^{\star})$ has only two possible cases: it either
intersects $\r_0\r_3$ and the polygonal line $\r_0\r_2\r_3$, or
intersects $\r_0\r_2$ and $\r_2\r_3$. In the first case, let the
intersection points of $\T(t^{\star})$ and $O(t_0)$, $O(t_1)$ be
$\r_{O_0}, \r_{O_1}$ respectively. Then $\r_{O_0}$ and $\r_{O_1}$
are on the same side with respect to $\r(t^{\star})$ in line
$\T(t^{\star})$. Which means that one of the sub-segments
$\r^{\star}_1(t),t\in[t_0,t^{\star}]$ and
$\r^{\star}_2(t),t\in[t^{\star},t_1]$ cannot be inside its
tetrahedron by the first paragraph of the proof, a contradiction to
Proposition~\ref{sub-prop}. In the second case, the points $\r_0$
and $\r_3$ are on the same side of $O(t^{\star})$. By
Proposition~\ref{sub-prop}, the sub-segment curves at $t=t^{\star}$
are also on the same side of $O(t^{\star})$. Then the curve $\r(t)$
does not pass through $O(t^{\star})$ at $t^{\star}$, which means
that $\tau(t^{\star})=0$ by Lemma~\ref{lemma-pre2}. Hence,
$\|\r_0\r_{01}\|$ and $\|\r_2\r_{23}\|$ are monotone.

It is known that $\r_{01}$ lies on $\r_0\r_1$ and  $\r_{23}$ lies on
$\r_2\r_3$. We claim that  $\r_{12}$ must be on $\r_1\r_2$.
Otherwise, assuming $O(t^{\star})$ has no common points with
$\r_1\r_2$, then $O(t^{\star})$ must intersect with
$\r_0\r_1,\r_0\r_2, \r_1\r_3,$ and $\r_2\r_3$. That means $\r_0$ and
$\r_3$ are on the same side of $O(t^{\star})$, and then
$\tau(t^{\star})=0$, a contradiction.

Since the curve is inside its tetrahedron, $\r(t^{\star})$ is inside
the quadrangle $\r_{01}\r_{12}\r_{23}\r_{30}$. Actually,
$\r(t^{\star})$ is inside the triangle $\r_{01}\r_{12}\r_{23}$.
$\r(t^{\star})$ cannot be on $\r_{01}$ and $\r_{23}$ according to
condition III).
 So, if $\r(t^{\star})$ is not inside the triangle
$\r_{01}\r_{12}\r_{23}$, then $\r(t^{\star})$ is on the opposite
side with $\r_{12}$ with respect to $\r_{01}\r_{23}$ or on
$\r_{01}\r_{23}$. Then $\T(t^{\star})\bigcap O(t_0)$ is not inside
$\r_{01}\r_{12}$, or, $\T(t^{\star})\bigcap O(t_1)$ is not inside
$\r_{12}\r_{23}$, since $\r_{01}\r_{12}\r_{23}\r_{30}$ is convex.
Without loss of generality, we suppose $\T(t^{\star})\bigcap O(t_0)$
is not in $\r_{01}\r_{12}$. Then $\T(t^{\star})\bigcap O(t_0)$ are
on the same side with $\r_2$ w.r.t. $\r_0\r_1$ in $O(t_0)$ by
Lemma~\ref{lemma-pre1}. Hence, $\T(t^\star)\bigcap O(t_0)$ and
$\T(t^{\star})\bigcap O(t_1)$ is on the same side of $\r(t^{\star})$
in $\T(t^{\star})$, which means that one of the sub-segments
$\r^{\star}_1(t),t\in[t_0,t^{\star}]$ and
$\r^{\star}_2(t),t\in[t^{\star},t_1]$ cannot be inside its
tetrahedron, a contradiction to Proposition~\ref{sub-prop}.

Therefore, $\r(t^{\star})$ can only be inside the triangle
$\r_{01}\r_{12}\r_{23}$, and $\T(t^{\star})$ can only intersect
$\r_{01}\r_{12}$ with $\r_{012}$ and intersect $\r_{12}\r_{23}$ with
$\r_{123}$. Subdivide $\r(t)$ at $t=t^{\star}$ to get curve segments
$\r_1^{\star}(t),t\in[t_0,t^{\star}]$, and
$\r_2^{\star}(t),t\in[t^{\star},t_1]$, and their tetrahedrons as
$\lozenge\r_0\r_{01}\r_{012}\r(t^{\star})$ and
$\lozenge\r(t^{\star})\r_{123}\r_{23}\r_3$. It has been shown that
these two sub-tetrahedrons are inside the tetrahedron
$\lozenge\r_0\r_1\r_2\r_3$. As a consequence, for any
$t^{\star}_1<t^{\star}_2$ in $[t_0,t_1]$, the sub-tetrahedron of the
sub-segment $\r^{\star}(t),t\in[t^{\star}_1,t^{\star}_2]$ is inside
the tetrahedron $\lozenge\r_0\r_1\r_2\r_3$.

Finally, we prove that $\|\r_1\r_{12}\|$ is monotone. It suffices to
show that there exist no $t^{\star}_1<t^{\star}_2\in[t_0,t_1]$ such
that $O(t^{\star}_1)$ and $O(t^{\star}_2)$ have a common point in
$\r_1\r_2$. Otherwise, we assume $O(t^{\star}_1)$ and
$O(t^{\star}_2)$ have a common point $\r_{12}^{\star}$ in
$\r_1\r_2$. Since $\r_0\r_{01}$ and $\r_2\r_{23}$ are monotonously
increasing, $\r_{01}(t_1^{\star})$ and $\r_{23}(t_1^\star)$ are on
the same side of $O(t_2^\star)$. Hence the intersection line of
$O(t^{\star}_1)$ and $O(t^{\star}_2)$ can only be outside of the
tetrahedron $\lozenge\r_0\r_1\r_2\r_3$ passing through
$\r_{12}^{\star}$. Then the sub-tetrahedron of the sub-segment
$\r^{\star}_{12}(t),t\in[t^{\star}_1,t^{\star}_2]$, cannot be inside
the tetrahedron $\lozenge\r_0\r_1\r_2\r_3$, which contradicts to the
consequence in the preceding paragraph. \hfill\qed\end{proof}

For clarity, we  summarize the properties mentioned in the proof of
the above theorem as follows.
\begin{proposition}\label{sub-inside}
For any $t^{\star}_1<t^{\star}_2\in[t_0,t_1]$, the sub-tetrahedron
$\lozenge\r_0^{\star}\r_1^{\star}\r_2^{\star}\r_3^{\star}$ of the
sub-segment $\r^{\star}(t),t\in[t^{\star}_1,t^{\star}_2]$ is inside
the tetrahedron $\lozenge\r_0\r_1\r_2\r_3$.
\end{proposition}

Similar to 7) of Lemma~\ref{Bezier-lm}, we have the following
proposition. The proof is also similar to that of 7) of
Lemma~\ref{Bezier-lm}.
\begin{proposition}\label{mono-sub}
$\|\r_0\r_{03}\|$ and $\|\r_1\r_{12}\|$ are strictly monotone with
$t^{\star}\in (t_0,t_1)$ where $\r_{03}$ and $\r_{12}$ are the
intersection points $\r_1\r_2\r(t^{\star})\bigcap \r_0\r_3$ and
$\r_0\r_3\r(t^{\star})\bigcap \r_1\r_2$ respectively.
\end{proposition}
\begin{proof}
It is sufficient to prove that the planes $\r_1\r_2\r(t^{\star})$
and $\r_0\r_3\r(t^{\star})$ are not tangent to $\r(t)$ at
$t^{\star}\in (t_0,t_1)$. If the plane $\r_1\r_2\r(t^{\star})$ is
tangent to $\r(t)$ at  $t^{\star}\in (t_0,t_1)$, then the
osculating plane $O(t^{\star})$ must intersect
$\r_1\r_2\r(t^{\star})$ with the tangent line $\T(t^{\star})$. By
Theorem~\ref{sub-thm}, $\T(t^{\star})$ must intersect $\r_1\r_2$
which is the common line of $O(t_0)$ and $O(t_1)$. Dividing the
curve segment into two sub-segments $\r_1^{\star}(t)$ and
$\r_2^{\star}(t)$, then one of them cannot be inside its
sub-tetrahedron according to Lemma~\ref{lemma-pre1} which
contradicts to~Proposition~\ref{sub-inside}. And one can similarly
discuss the case for the plane $\r_0\r_3\r(t^{\star})$.
\hfill\qed\end{proof}

According to Proposition~\ref{mono-sub}, $\r(t)$ and the plane
$\r_1\r_2\r_M$ have a unique intersection point $\s_\r$ where
$\r_M=(\r_0+\r_3)/2$. We call $\s_\r$ the \emph{shoulder point} of
the segment $\r(t),t\in[t_0,t_1]$.
Similar to Proposition~\ref{sub-prop}, we can see that
Theorem~\ref{sub-thm} and Proposition~\ref{mono-sub} also hold for
any subsegment $\r^{\star}(t),t\in[t^{\star}_1,t^{\star}_2]$.

When we subdivide
the approximated curve segment at a point $t=t^{\star}$, by
Theorem~\ref{sub-thm}, we assume that the osculating plane
$O(t^{\star})$ intersects $\r_0\r_1, \r_1\r_2$ and $\r_2\r_3$ at
$\r_{01},\r_{12}$ and $\r_{23}$ respectively.
%
Then, one can have the following corollary.

\begin{corollary}\label{prop-in}
Let $k_1(t^\star)=\frac{|\r_1\r_{01}|}{|\r_1\r_0|},
k_2(t^\star)=\frac{|\r_2\r_{12}|}{|\r_2\r_1|}$ and
$k_3(t^\star)=\frac{|\r_3\r_{23}|}{|\r_3\r_2|}$,
 then $k_i(t^\star)$ is monotone and $k_i(t^\star)\in (0,1)
$ with $t^\star\in (t_0,t_1)$, $i=1,2,3$.
\end{corollary}

We finally give the \textbf{Proof of Theorem~\ref{quasi-cubic}} by summarizing the above discussions.
\begin{proof} Set $t_1$ as Theorem~\ref{divide-thm}, then $\r(t),t\in[t_0,t_1]$ has the eight properties in Theorem~\ref{divide-thm}, \ref{sub-thm} and  Propositions~\ref{sub-inside}, \ref{mono-sub}.  It means that  the segment $\r(t),t\in[t_0,t_1]$ is a quasi-cubic segment.
\hfill\qed\end{proof}

\subsection{Subdivision algorithm}
\label{sec-sa3}

As we mentioned in the introduction, the topology graph
$\mathcal{G}$ of a parametric space curve can be computed by the
method in~\cite{Alcazar2009}.

A {\em topology graph} is a graph $\mathcal{G}= \{\PV, \PE\}$ where
$\PV$ is a set of points in the Euclidean space $\PV =
\{\vv_{i}=(\alpha_i,  \beta_i,\gamma_i)\}$ and $\PE$ is a set of
edges $\PE=\{(\vv_i, \vv_j) | \vv_i, \vv_j\in\PV\} $, any two edges
do not intersect except in the endpoints. A graph $\mathcal{G}$ is
a topology graph of a parametric space curve $\r(t)$ if
$\mathcal{G}$ and $\r(t)$ have the same topology.

The singular points of the space curve are included as vertices in
$\mathcal{G}$. In this paper, we need to add more information to the
vertices in our algorithm. For each vertex $\vv_i$ in the topology
graph, we now update it to
\begin{eqnarray}\label{vertex}
  V_i&=&\{\vv_i=\r(t_{i0}),\{t_{i0},t_{i1},\dots,t_{ik}\}, \nonumber \\
  & & \{\mathcal{F}_{i0}^-,\dots,\mathcal{F}_{ik}^-\},  \{\mathcal{F}_{i0}^+,\dots,\mathcal{F}_{ik}^+\}\},
\end{eqnarray}
where each $t_{ij}$ is a real parameter such that
$\r(t_{ij})=\vv_i$, $\mathcal{F}_{ij}^-$ and $\mathcal{F}_{i0}^+$
are the left and right Frenet frames of $\vv_i$ with respect to the
parameters $t_{ij},j=0,\ldots,k$.
The point set $\PV$ thus updated is called the {\em extended vertex
list}. Methods to compute the limitation of the tangent are also
introduced in~\cite{Daouda2008}.

The edges in~$\mathcal{G}$ are not used directly in our
approximation algorithm, but they give the connection relationship
of two updated vertices.
Since the space curve is parametric, the connection relationship is
given by the parameters corresponding to the points in $\PV$ in the
increasing order.
%
%
So in our paper, we use the extended vertex list $\mathcal{V}$
instead of topology graph.

\begin{example} Figure~\ref{fig01} (a) shows a space curve with a cusp, whose
topology graph is given in Figure~\ref{fig01} (b).
Figure~\ref{fig02}(a) shows a numerical approximate curve which does
not pass through the cusp. We may use the topology graph or a
refined topology graph to approximate the curve segment as shown in
Figure~\ref{fig02}(b). This method has two drawbacks. First, we
generally needs hundreds even thousands line segments to approximate
the curve segment for a small precision \cite{curve-ib}.  Second,
the approximate curve cannot keep the tangent directions of left
and right sides of the cusp point.
In this paper, we use a cubic B\'{e}zier curve instead of a line segment
as shown in Figure~\ref{fig02}(c), which is not only more precise
but keeps the geometric properties of the original curve.
  \begin{figure}[!h]
 \centering
 \includegraphics[width=0.3\textwidth]{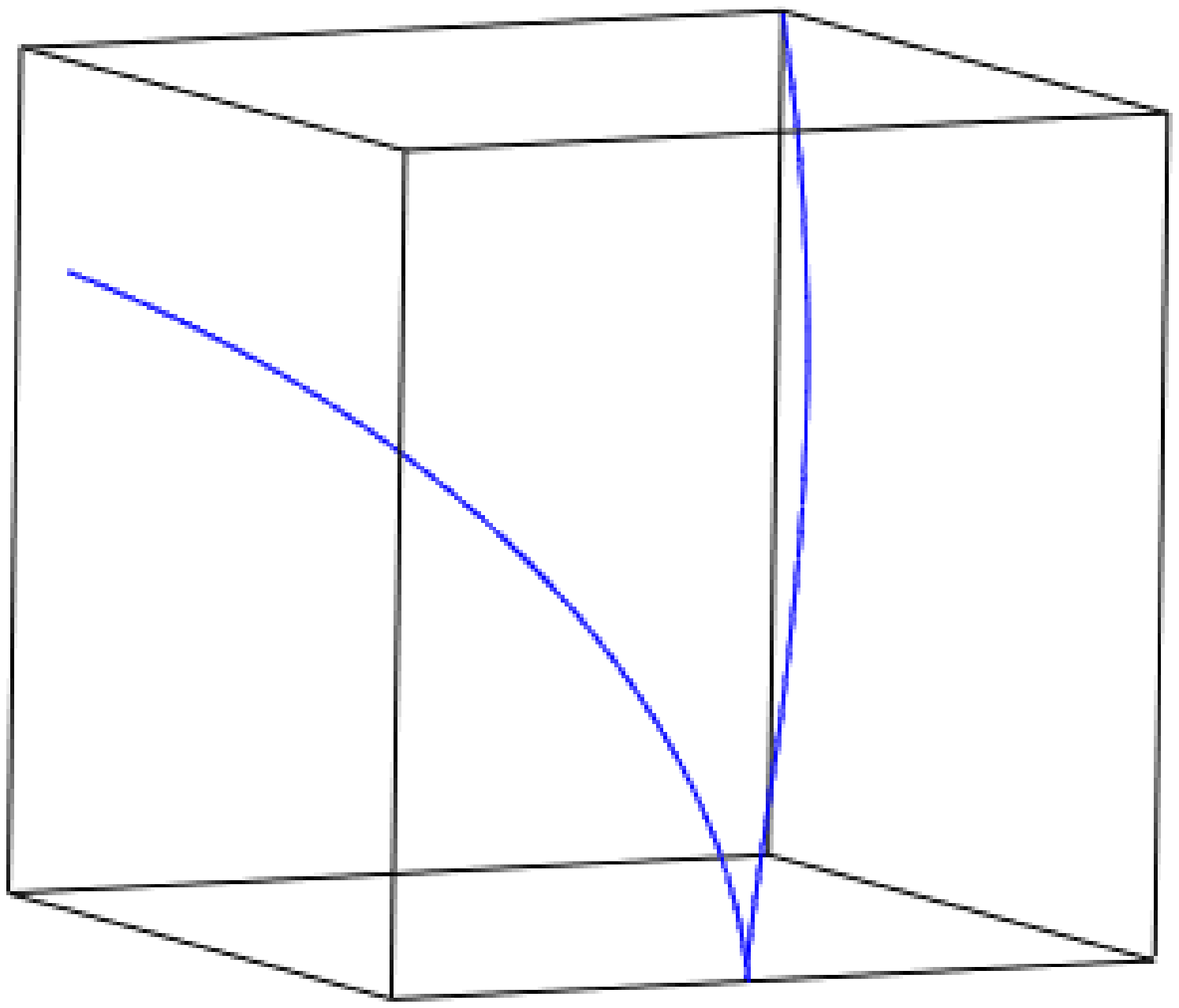}\qquad
 \includegraphics[width=0.3\textwidth]{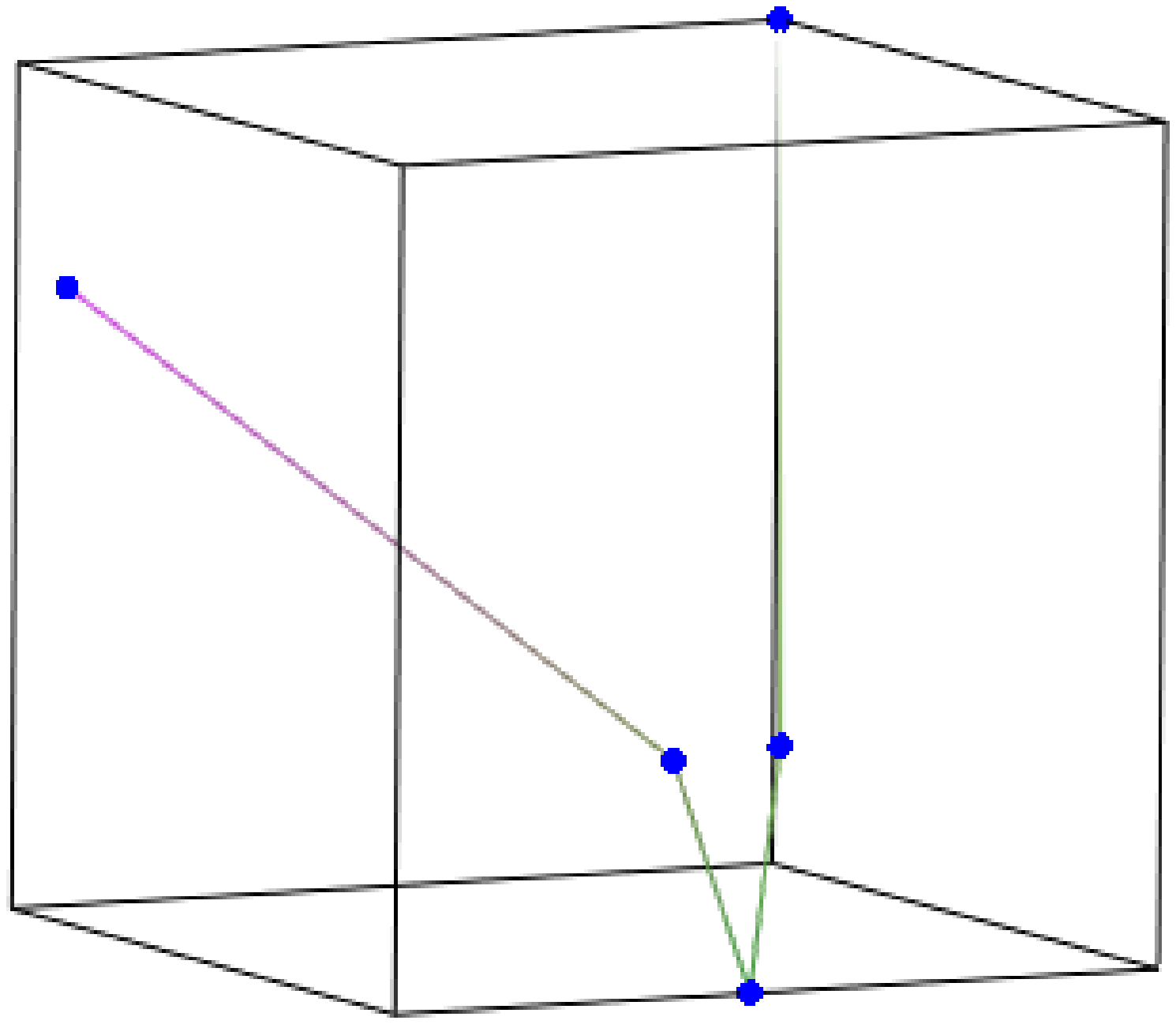}\\
 {  \scriptsize  (a) Origin curve \hspace{1.2in}    (b) Topology graph
}
  \caption{Topology graph of the curve}
  \label{fig01}
 \end{figure}

  \begin{figure}[!h]
 \centering
  \includegraphics[width=0.3\textwidth]{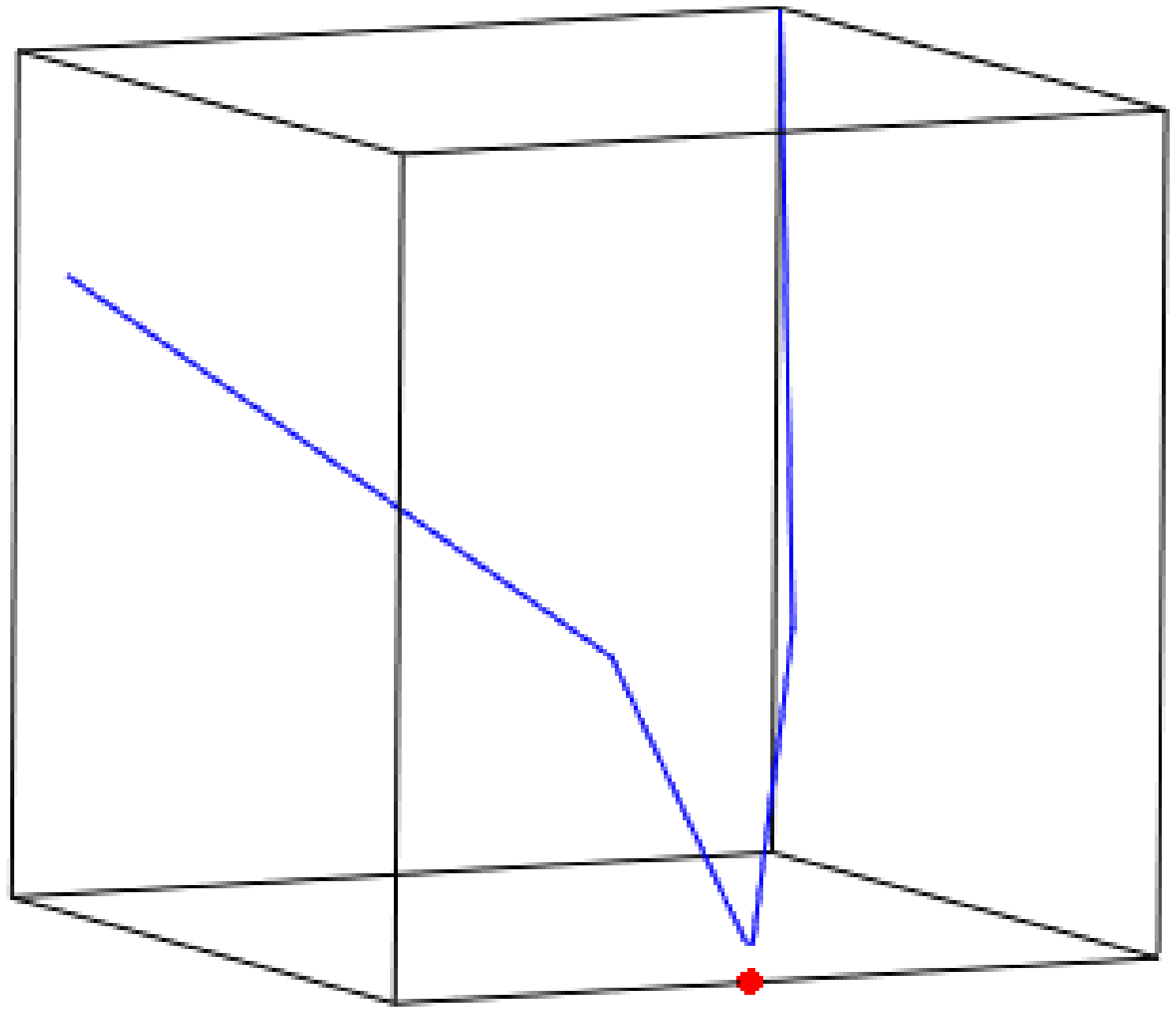}\quad
   \includegraphics[width=0.3\textwidth]{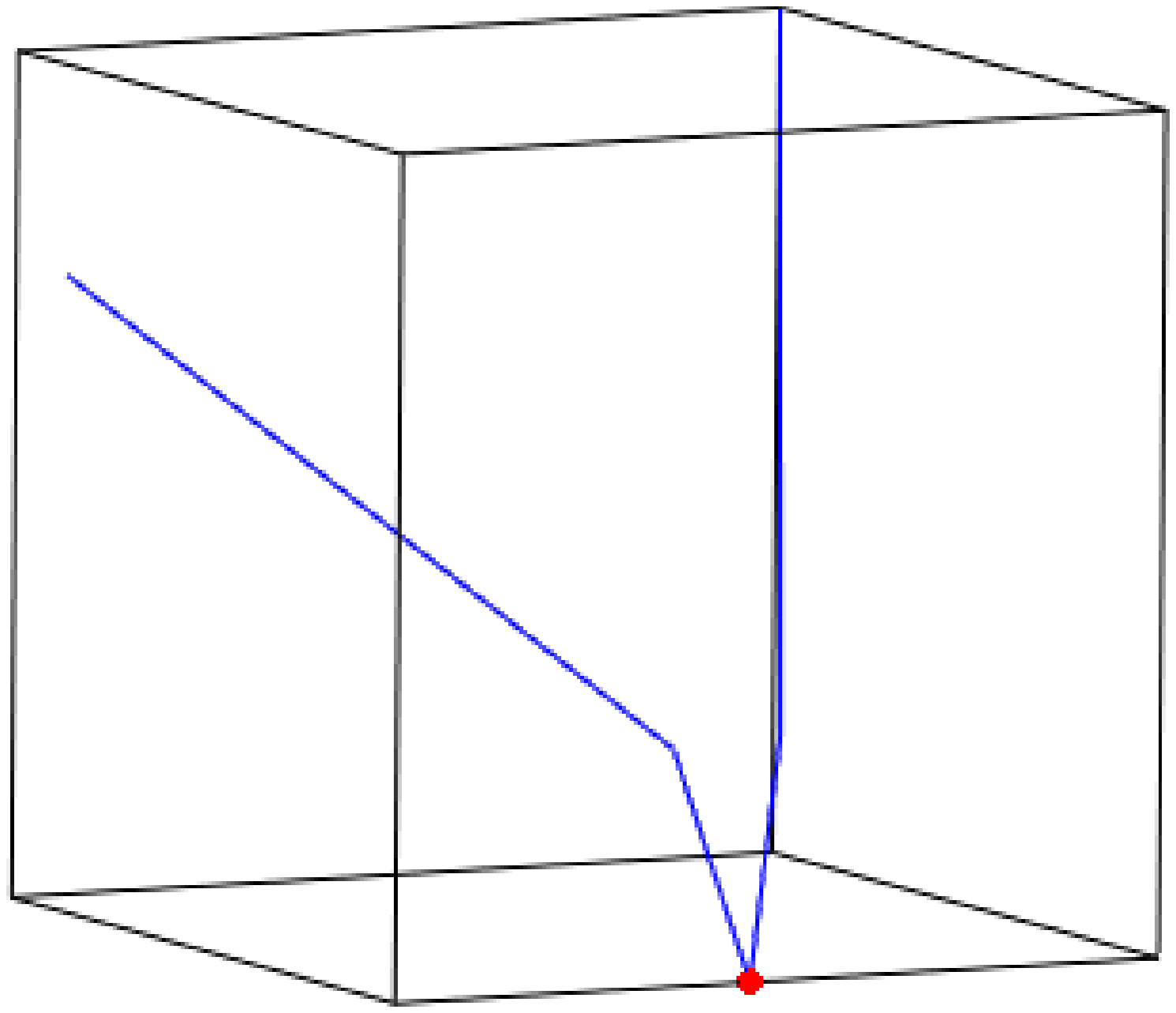}\quad
    \includegraphics[width=0.3\textwidth]{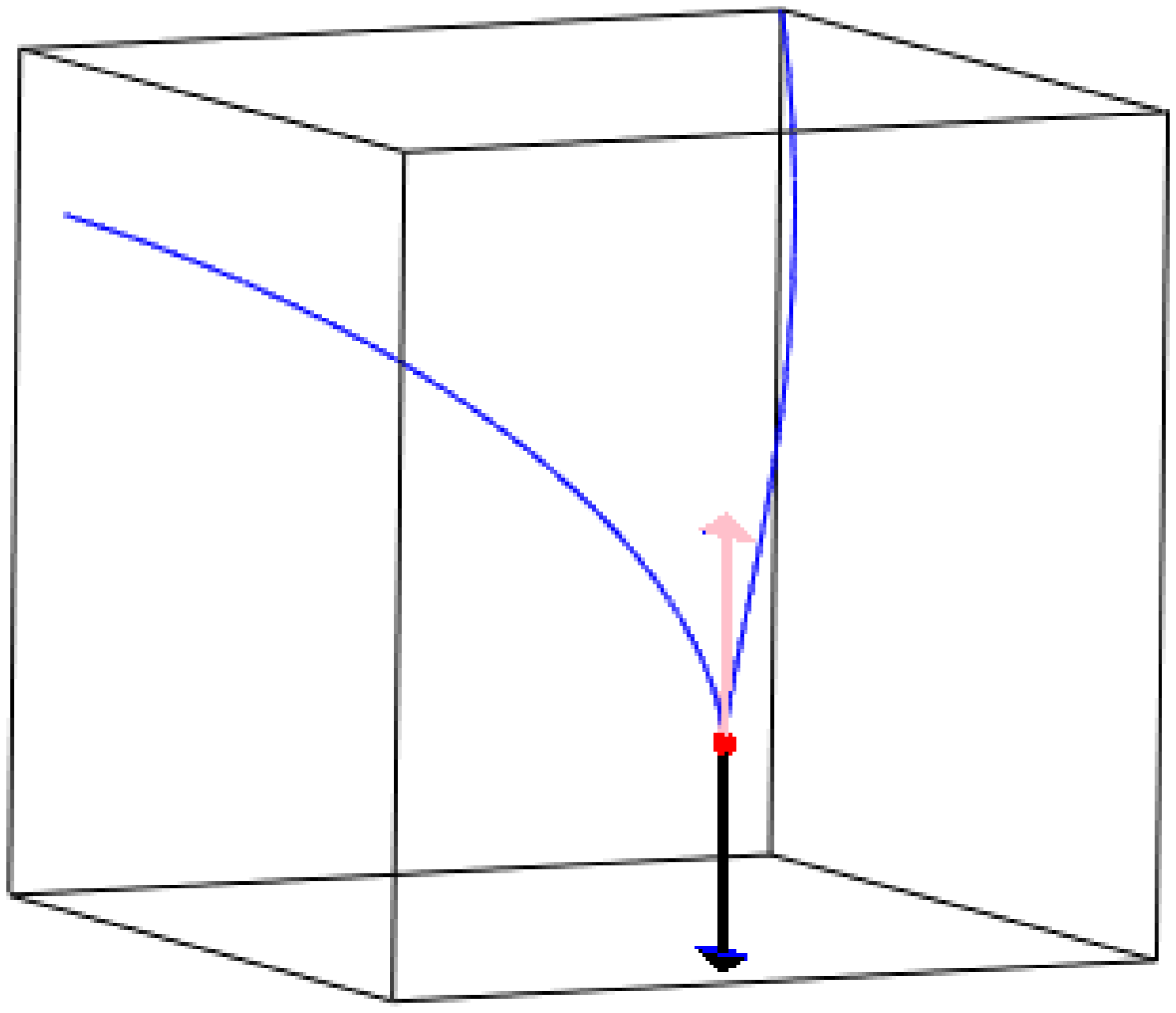}\\
{\scriptsize \quad   (a) General numerical method \hspace{0.4in}   (b) Based on
topology \hspace{0.5in}  (c) Proposed method}
 \caption{Numerical approximate curve}
 \label{fig02}
\end{figure}
\end{example}

Based on the above analysis, we now give the segment dividing
algorithm.
\begin{alg}\label{divide-alg} Curve Subdivision.\\
  \textbf{Input}: A normal curve segment $\r(t),t\in [0,1]$.\\
  \textbf{Output}: An extended vertex list with elements as~\eqref{vertex}.
\end{alg}
  \begin{enumerate}
  \item Compute the certified vertex list $\mathcal{V}$ with all
  character points as vertices with the method in \cite{Alcazar2009}.
  The parameters and  the left  and right Frenet frames are recorded.
  Suppose the real roots associated to the character points are $s_i,i=1,\ldots,l-1$ and
 $0=s_0<s_1<\cdots<s_l=1$.

    \item Divide each interval $[s_i,s_{i+1}]$ as
    $s_i=s_{i0}<s_{i,1}<\cdots<s_{i,k_i}=s_{i+1}$ such that
    each segment satisfies the conditions given in I) to IV).
   \item Rearrange the $s_{ij}$ in an ascending order and rename them as
   $t_i,i=0,\ldots,n$. Find the left and right Frenet frames of each segment
    $\r(t),t\in[t_i,t_{i+1}]$.
    \item Add all these new points to the extended vertex list $\mathcal{V}$
    which is now ready for approximation.
  \end{enumerate}

Each curve segment is defined by two adjoint vertices of
$\mathcal{V}$. By Proposition~\ref{sub-prop}, the curve segment from
the algorithm is in the tetrahedron and has the properties in
Theorems~\ref{divide-thm}, \ref{sub-thm} and
Propositions~\ref{sub-inside}, \ref{mono-sub}. Hence each curve
segment obtained from Algorithm \ref{divide-alg} is a quasi-cubic
segment and so are its sub-segments.

\section{Shoulder point approximation}
In this section, we propose an efficient algorithm to construct a
set of cubic B\'{e}zier curve segments which approximate a quasi-cubic
segment obtained in Algorithm~\ref{divide-alg} to any approximate
bound.

Firstly, we focus on one quasi-cubic segment $\r(t),t\in[t_0,t_1]$.
Let $\r_0,\r_3$ be the endpoints of the segment, $\r_1$ the
intersection point of the tangent line at $\r_0$ and the osculating
plane of $\r_3$, and $\r_2$ the intersection point of the tangent
line at $\r_3$ and the osculating plane of $\r_0$. Then
$\{\r_0,\r_1,\r_2,\r_3\}$ defines a family of rational cubic curves
\begin{equation}\label{Bezier-construct}
\p(\omega_1,\omega_2,s)=\frac{\r_0B_0(s)+\omega_1\r_1B_1(s)+\omega_2\r_2B_2(s)
+\r_3B_3(s)}{B_0(s)+\omega_1B_1(s)+\omega_2B_2(s)+B_3(s)},\ s\in
[0,1].
\end{equation}
Then $\p(\omega_1,\omega_2,s)$ is called the \emph{associated cubic B\'{e}zier curve segment} of $\r(t)$.
It has been shown that $\p(\omega_1,\omega_2,s)$ meets $\r(t)$ at
its endpoints $\r(t_0)$ and $\r(t_1)$. Furthermore,
$\p(\omega_1,\omega_2,s)$ and $\r(t)$ have the same left and right
tangent directions and osculating planes at the endpoints, and the
same control tetrahedron $\lozenge\r_0\r_1\r_2\r_3$.

\begin{proposition}\label{order}
Let $\p(\omega_1,\omega_2,s),s\in[0,1]$ be the associated cubic B\'{e}zier curve segment of $\r(t),t\in[t_0,t_1]$. Then $\p(\omega_1,\omega_2,s)$ can approximate $\r(t)$ at their endpoints with order two by setting proper $\omega_1$ and $\omega_2$, i.e., $\{\p(0)=\r(t_0),\p(1)=\r(t_1)\}$ and $\{\p'(0)=\r'(t_0),\p'(1)=\r'(t_1)\}$.
\end{proposition}
\begin{proof} Following the construction of $\p(s)$ for $\r(t)$, they are $G^1$ interpolated at their endpoints with arbitrary weights $\omega_1$ and $\omega_2$. According to the properties of the cubic B\'{e}zier curve, one can set the proper $\omega_1$ and $\omega_2$ such that $\p(s)$ and $\r(t)$ are $C^1$ interpolated at their endpoints.
\hfill\qed\end{proof}

In Proposition~\ref{order}, the weights are selected to enhance the approximation order from $G^1$ to $C^1$ at the endpoints.
Actually, on can get $\{\p(\omega_1,\omega_2,0)=\r(t_0),\p(\omega_1,\omega_2,1)=\r(t_1)\}$ and $\{\p'(\omega_1,\omega_2,0)=k_1\omega_1\r'(t_0),\p'(\omega_1,\omega_2,1)=k_2\omega_2\r'(t_1)\}$,
 where $k_1$ and $k_2$ are positive constants. Hence we can set $\omega_1$ and $\omega_2$ such that $k_1\omega_1=1$ and $k_2\omega_2=1$.
However, in the following paragraphs, we would like to use the freedom of weights to minimize the position approximation error. Hence, we will
show how to compute the proper weights $\omega_1,\omega_2$ such that
$\p(s)$ is an optimal approximation to $\r(t)$.

The selection of the weights often leads to some optimization
problems such as $\min_{\omega_1,\omega_2}$ $(\max_{s,t}
d(\omega_1,\omega_2,s,t)^2)$ where $d(\omega_1,\omega_2,s,t)$ is the
distance function between $\p(\omega_1,\omega_2,s)$ and $\r(t)$ in
certain forms~\cite{Pottmann2002}. The computation is usually not
efficient and some global error analysis is introduced to simplify
the optimization problem~\cite{Dooken2001}.
Another possible method is to approximate the target curve segment
by checking the parallel points. We can push the parallel points of
the approximated curve and the approximate
curve~\eqref{Bezier-construct} as near as possible. It also leads to
an optimal problem for a function with degree three. In the
following, we introduce a novel method which avoids any
optimizations.

The shoulder point $\s_\p$ of $\p(s)$ is given in
Proposition~\ref{shoulder}. The shoulder point $\s_\r$ of $\r(t)$ can be
computed as the unique intersection point of $\r(t)$ and the
triangle $\r_1\r_2\r_M$. Supposing the plane $P(x,y,z)$ is defined by
$\r_1$, $\r_2$, and $\r_M$, then the shoulder point corresponds to a
real root $t^{\star}\in (t_0,t_1)$ of $P\circ r(t)$ with
$\r(t^{\star})$ lying in the triangle  $\r_1\r_2\r_M$. So
$D(\omega_1,\omega_2)=\|\s_\p-\s_\r\|^2$ is a rational function in
$\omega_1,\omega_2$ with total degree two. Finding the positive
solution from the equations
\begin{equation}\label{weight_equation}
\left\{
\begin{array}{l} \dfrac{\partial{D}}{\partial
\omega_1}=0,\\[0.3cm]
\dfrac{\partial{D}}{\partial \omega_2}=0,
\end{array}
\right.
\end{equation}
we obtain the weights for
the approximate cubic curve~\eqref{Bezier-construct}.

Before the approximation, we will estimate the error between the two
curves. Since there does not have any simple method to compute the
distance of two parametric curves with different parameters, we use
the distance between $\r$ and the implicit variety of a rational
cubic curve $\p$. It has been proved that the associated implicit
ideal $I_\p$ of $\p$ can be computed using the $\mu$-basis
method~\cite{Cox1998} efficiently:
%
\begin{lemma}\label{mu-imp}The associated ideal of $\p$ has the form $I_\p=\langle f(x,y,z),g(x,y,z)$, $h(x,y,z)\rangle$, where $f,g$ and $h$ are quadratic
polynomials, i.e., the resultants of $\p's$ $\mu$-basis in pairs.
\end{lemma}
The algorithm of $\mu$-basis is given in~\cite{deng2005}.
Generalizing the approximation error function in~\cite{Chuang1989},
we have
$$e(f,\r)=\left(\frac{f(\r)^2}{f_x(\r)^2+f_y(\r)^2+f_z(\r)^2}\right)^{1/2}.$$
Let $e(\p,\r):=e(f,\r)+e(g,\r)+e(h,\r)=e(t)$ be the univariate error
function in $t$. Then the approximation error can be set as the
following optimization problem:
$$e=\max_{t_0\le t\le t_1}(e(t)).$$
There are many methods to solve this problem. However, for the
efficiency in practice, we often sample $t$ as
$t_i=\frac{(t_1-t_0)i}{m},i=0,\ldots,m,$ for a proper $m$, say
$m=300$, and set the approximate error as $\max(e(t_i))$.

The following algorithm is proposed to approximate a quasi-cubic
curve segment via shoulder point approximation.
\begin{alg}\label{approx-alg} Shoulder point approximation\\
\textbf{Input}: A quasi-cubic curve segment $\r(t),t\in[t_0,t_1]$
 and a positive error bound
  $\delta$.\\
\textbf{Output}: A set of cubic B\'{e}zier curves which is a
$\delta$-approximation for $\r(t)$.
\end{alg}
\begin{enumerate}
\item[1.] Construct the associated tetrahedron of $\r(t)$ and the rational B\'{e}zier
cubic curve $\p(\omega_1,\omega_2,s),s\in[0,1]$ as shown in
\eqref{Bezier-construct}.
   \item[2.] Compute the weights $(\omega_1,\omega_2)$ such that $\|\s_\p-\s_\r\|$
        is as small as possible.
  \begin{enumerate} \item Compute shoulder points $\s_\r$ and $\s_\p(\omega_1,\omega_2)$
       of $\r(t)$ and $\p(s)$ respectively.
    \item  Find a pair of real roots $(\omega_1,\omega_2)$ by solving the equation system~\eqref{weight_equation}.
      \end{enumerate}
    \item[3.] Compute the approximate error $\bar\delta=e(t)$. If $\bar\delta<\delta$
    then output $\p(s)$.
     Otherwise, divide $\r(t)$ to two parts on its middle point of arc length and repeat the approximation process for each subsegment.
  \end{enumerate}

%

%
\begin{example}
A curve segment $\r(t), t\in[0,21/32]$ represented by the black
curve with degree six is given by
Algorithm~\ref{divide-alg} and the approximate cubic B\'{e}zier curve is
the red dash curve in Figure~\ref{fig2}. The weights are
$\omega_1=\omega_2=1$ in the left figure. After executing step 2 of
Algorithm~\ref{approx-alg}, we have $\omega_1=5/11,\omega_2=16/31$
in the right figure. The numerical errors are $0.29$ and $0.04$
respectively computed from error function $e(t)$ by setting $m=300$.
  \begin{figure}[!h]
 \centering
 \includegraphics[width=0.40\textwidth]{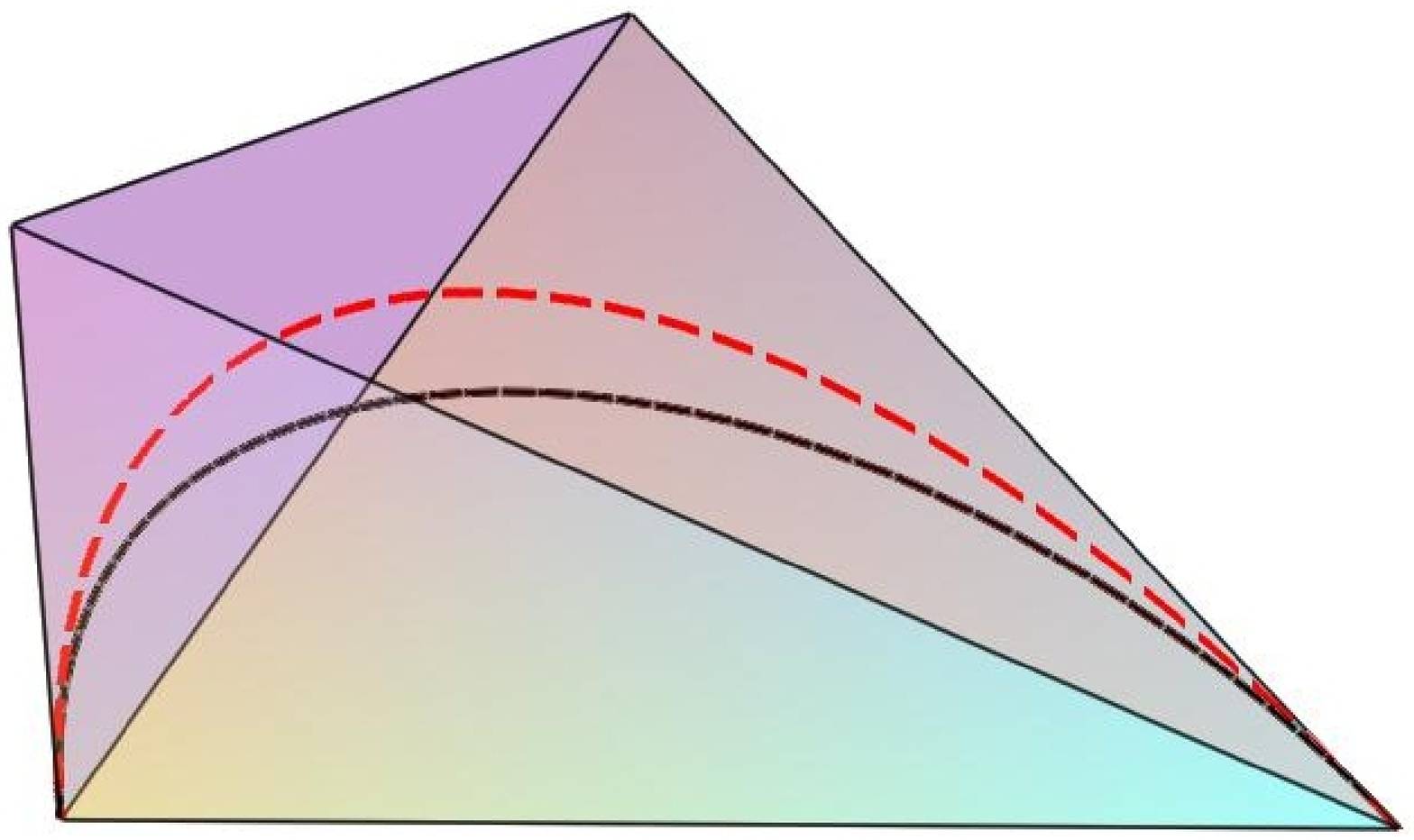}\qquad
 \includegraphics[width=0.41\textwidth]{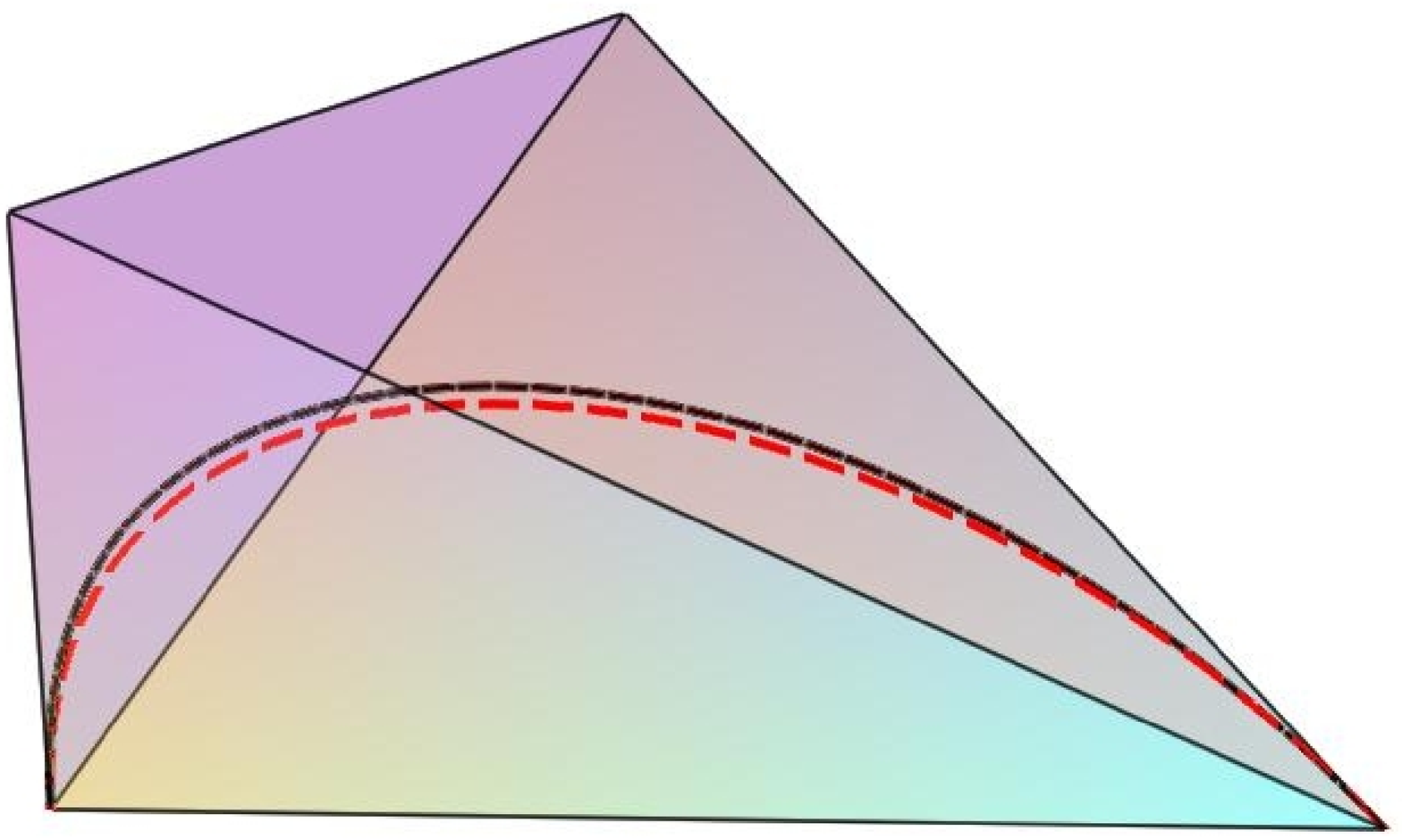}
 \caption{Selecting the weights for B\'{e}zier cubic curve}
 \label{fig2}
\end{figure}
\end{example}

To show the termination of the above algorithm, we need the following lemma.
\begin{lemma}\label{lm-term}
The edge of the sub-tetrahedron in Algorithm~\ref{approx-alg}
converges to zero when the arc length of its subdivided curve
segment converges to zero.
\end{lemma}
\begin{proof}
There exists a $t=t^\star_1\in (t_0,t_1)$ such that $k_1=1/2$ since
$k_1(t)$ is monotone with $t$ in $(t_0,t_1)$ by
Corollary~\ref{prop-in}. Consider the subsegment $\r(t),t\in
[t_0,t^\star_1]$ and subdivide it at $t=t^\star_2$ such that
$k_2=1/2$ for the sub-tetrahedron $\lozenge(t_0,t^\star_1)$. Then,
subdivide $\r(t),t\in [t_0,t^\star_2]$ at $t=t^\star_3$ such that
$k_3=1/2$ for $\lozenge(t_0,t^\star_3)$. Let $t^{(1)} = t^\star_3$.
We obtain a subsegment $\r(t),t\in [t_0,t^{(1)}]$ whose
sub-tetrahedron $\lozenge(t_0,t^{(1)})$ has vertices
$\r_0^{(1)}=\r_0,\r_1^{(1)},\r_2^{(1)},\r_3^{(1)}$.
Similarly, we can construct $\r_j^{(i)}, j=0,1,2,3$ and $t^{(i)}$.
According to the subdividing process, let $\r_j^{(0)} = \r_j,
j=0,1,2,3$.
Then, we have
$\|\r_0^{(i)}\r_1^{(i)}\| < \|\r_0^{(i-1)}\r_1^{(i-1)}\|/2$,
$\|\r_1^{(i)}\r_2^{(i)}\| <
\|\r_0^{(i-1)}\r_1^{(i-1)}\|/2+\|\r_1^{(i-1)}\r_2^{(i-1)}\|/2$ and
$\|\r_2^{(i)}\r_3^{(i)}\|<
\|\r_0^{(i-1)}\r_1^{(i-1)}\|/2+\|\r_1^{(i-1)}\r_2^{(i-1)}\|+\|\r_2^{(i-1)}\r_3^{(i-1)}\|/2$
for $i>0$. Hence, the lengthes of the three edges
$\|\r_0^{(i)}\r_1^{(i)}\|$, $\|\r_1^{(i)}\r_2^{(i)}\|$ and
$\|\r_2^{(i)}\r_3^{(i)}\|$ of a sub-tetrahedron
$\lozenge(t_0,t^{(i)})$ converge to zero when $i\rightarrow
\infty$. Since $\r(t),t\in [t_0,t_1]$ is a rational curve and has no
singular point, $t^{(i)}-t_0$ converges to zero when
$i\rightarrow \infty$.
%

Let $t\in[t_0,t_1]$ and $\lozenge \r_0\r_1'(t)\r_2'(t)\r_3'(t)$ its
tetrahedron. Then $s(t) = \|\r_0\r_1'(t)\| + \|\r_1'(t)\r_2'(t)\| +
\|\r_2'(t)\r_3'(t)\|$ converges to zero when $t\rightarrow t_0$,
since $\r(t)$ has no singularities in $[t_0,t_1]$.
%
Hence when the
arc length of its subdivided curve segment converges
to zero, which means $t\rightarrow t_0$, the edge of sub-tetrahedron converges to zero.
\hfill\qed\end{proof}

The termination of Algorithm~\ref{approx-alg} can be guaranteed by
the following theorem.

\begin{theorem}\label{convergent}
  In Algorithm~\ref{approx-alg}, the approximation error converges to zero for the
  subdivision procedure.
\end{theorem}
\begin{proof}
By Lemma~\ref{lm-term}, when the arc length of its subdivided curve
segment converges to zero, the edge of the sub-tetrahedron
converges to zero. Since the approximation error is controlled by
the edges, it converges to zero for the subdivision procedure.
\hfill\qed\end{proof}

  \begin{remark}
  In Algorithm~\ref{approx-alg}, the Step 3 is given to simplify the
  proof of the convergence. In fact, for less computation, we always implement the algorithm
  with the following step instead of $3$.
  \begin{enumerate}
    \item[$3'$.] Compute the approximate error $\bar\delta=e(t)$. If $\bar\delta<\delta$
    then output $\p(s)$.
     Otherwise, divide $\r(t)$ to two parts on its shoulder point $\s_\r$ repeat the approximation process for each subsegment.
  \end{enumerate}
   According to the proof of
  Lemma~\ref{lm-term}, the algorithm fails if a subsequence of $s_i$ does not converge to zero under shoulder
  point subdivision process,  and it never happened in our experiments. It is an interesting problem to prove the termination of this version of the algorithm.

\end{remark}

\section{Algorithms and experimental results}
After dividing the curve to segments by Algorithm~\ref{divide-alg},
we can approximate each curve segment by the shoulder approximation
method in Algorithm~\ref{approx-alg}. In this section, we give the
main approximation algorithm and the experimental results.

The global approximation is based on the local approximation and topology determination in the above sections. Some relationships of the approximate curve segments are considerable in the global view. In our approximation, the line edges in the topology graph are replaced by the associated cubic B\'{e}zier curve segments. To ensure the topological isotopy before and after the replacement, we restrict the cubic curve segments to have the appropriate topology based on the topology graph.

It is shown that an associated cubic B\'{e}zier curve segment are decided by its tetrahedron. Let $\lozenge\p_0^1\p_1^1\p_2^1\p_3^1$ and $\lozenge\p_0^2\p_1^2\p_2^2\p_3^2$ be two control tetrahedrons of two cubic B\'{e}zier curve segments $\p^1(s)$ and $\p^2(s)$. Then $\p^1(s)$ and $\p^2(s)$ can have no common points except for their endpoints.
In the further consideration, we give two cases for the problem. The first case is that $\p^1(s)$ and $\p^2(s)$ have only one common point being the endpoint and the same Frenet frames at this endpoint. And the other positional situations of $\p^1(s)$ and $\p^2(s)$ are included in the second case.

If all the pairs of cubic B\'{e}zier curves satisfy the second case, then to ensure that cubic curve segment does not bring in the unexpected knots while it replaces the line edge, one can give a sufficient condition that each cubic curve segment has no common points with the control tetrahedron of another curve segment except for the endpoint. This condition can be strengthened if we do not want to check the collision between a cubic curve segment and a tetrahedron. The condition can be that the two tetrahedrons have no inner points.
By Lemma~\ref{lm-term}, the condition can be satisfied by subdividing the curve segments.
Then the approximate curve have same topology with the given curve, since the approximate curve is controlled by the sequence of the tetrahedrons. Each tetrahedron has no common inner points with other tetrahedrons.

We then only need to discuss the pairs of cubic B\'{e}zier curves belong to the first case.
Assuming $\p_0^1=\p_0^2$, then $\p_1^2$ is on the radial $(1-\lambda)\p_0^1+\lambda\p_1^1,\lambda\ge 0$, and $\p_2^2$ is on the same side with $\p_2^1$ on the plane $\p_0^1\p_1^1\p_2^1$. According to the monotonicity of the B\'{e}zier curve in Lemma~\ref{Bezier-lm}, $\p^1(s)$ and $\p^2(s)$ can replace the their associated line edges without topology modification.

\begin{alg}\label{final-alg} Certified B-spline approximation with error
bound.\\
\textbf{Input}: A normal curve segment $\r(t),t\in [t_0,t_1]$ and a
positive error bound $\de$.\\
\textbf{Output}: A cubic B-spline $\p(s)$ such that the approximate
error between $\p(s)$ and $\r(t)$ is less than $\de$ and the
 approximate implicit spline for $\r(t)$.
\end{alg}
\begin{enumerate}
\item Divide the curve $\r(t)$ into quasi-cubic segments by
Algorithm~\ref{divide-alg}.
\item Check the topology conditions.
    \begin{enumerate}
      \item Check the intersection of any pair of cubic B\'{e}zier curves which have the same Frenet frame at the endpoint, divide them to two parts on their shoulder points respectively, if they have common points more the endpoints.
      \item Check the collision of any pair of tetrahedrons, divide them to two parts on their shoulder points respectively, if they have inner points.
    \end{enumerate}
\item For each segment, find the cubic B\'{e}zier curves which
 approximate the given curve segment with precision $\de$ by Algorithm~\ref{approx-alg}.
 \item Find the implicit form for the cubic B\'{e}zier curves with the $\mu$-basis method~\cite{Cox1998}.
 \item Convert the resulting rational cubic B\'{e}zier  curves to a
     rational B-spline with a proper knot selection as the method presented
     in~\cite{Piegl1997}.
  \end{enumerate}
\begin{remark}
In the process of topology conditions checking, we only need to check the collision of the sub-tetrahedrons subdivided from which are the intersected before the subdivision, since the sub-tetrahedrons are included in its father tetrahedrons. It means that the less and less pairs of tetrahedrons need to be checked in the subdivision process.
\end{remark}

\begin{theorem} From Algorithm~\ref{final-alg}, we obtain a piecewise $C^1$ continuous
approximate cubic B-spline curve which keeps the singular points,
inflection points, and torsion vanishing points of the approximated
parametric curve. At cusps, the approximate curve is $C^0$
continuous.
\end{theorem}
\begin{proof}
Algorithm~\ref{final-alg} gives the $G^1$ cubic B\'{e}zier spline since
it is constructed as the hermite interpolation of the original
curve, if the character points are not cusps. Then $C^1$ continuity
can be ensured from the conversion from the B\'{e}zier spline with a
proper knot selection~\cite{Piegl1997}. The singular points
of the curve are treated as segmenting points. Since at the
segmenting points, the left and right Frenet frames are preserved,
the origin curve and the approximate curve have the same singular
points. Since the cubic spline introduces no more singular points,
the algorithm keeps the singular points. At a cusp, its left (right)
tangent and osculating plane are kept according to
Algorithm~\ref{divide-alg}, and the approximate curve is then only
$C^0$ continuous.

The character points include the vertices of the topology graph. The topology conditions ensure that the topology is persevered while the topolgy line edges are replaced by the cubic B\'{e}zier curve segments.
According to Theorem~\ref{convergent}, the approximate curve from
Algorithm~\ref{final-alg} converges to the approximated curve
and they have the same topology.

The left and right Frenet frames of the approximate curves are the
same as that of the approximated curve at the character points,
which means that the principal normal vector and the osculating
plane are both kept. Then the principal normal vector changes its
direction at the inflection point. Similarly, the curve does not
pass through the osculating plane at the torsion vanishing point.
\hfill\qed\end{proof}

Finally, we give several examples to illustrate the algorithm.
\begin{example}
The space curve  $\r_1(t)$ from Example 6 in~\cite{Alcazar2009} has
a singular point $(0,0,0)$ at $t=\pm 1,\pm \infty$, where
  $$\r_1(t)=\left({\frac {1-t^{2}}{ \left( t^{2}+1 \right) ^{2}}},{\frac {t \left(
1-t^{2} \right) }{ \left( t^{2}+1 \right) ^{2}}},{\frac {t^{2}
 \left( 1-t^{2} \right) }{ \left( t^{2}+1 \right) ^{4}}}\right).$$
The curve segment $\r_1(t),t\in [-2,2]$ and its approximate spline
curve $\p(s)$ are shown in Figure~\ref{fig3}, they are shown in the same figure for comparison and the tetrahedron sequence is also given in Figure~\ref{fig31}, the numerical error
$e(t)$ is shown in Figure~\ref{fig30}.
 \begin{figure}[!htp]
 \centering
 \includegraphics[width=0.35\textwidth]{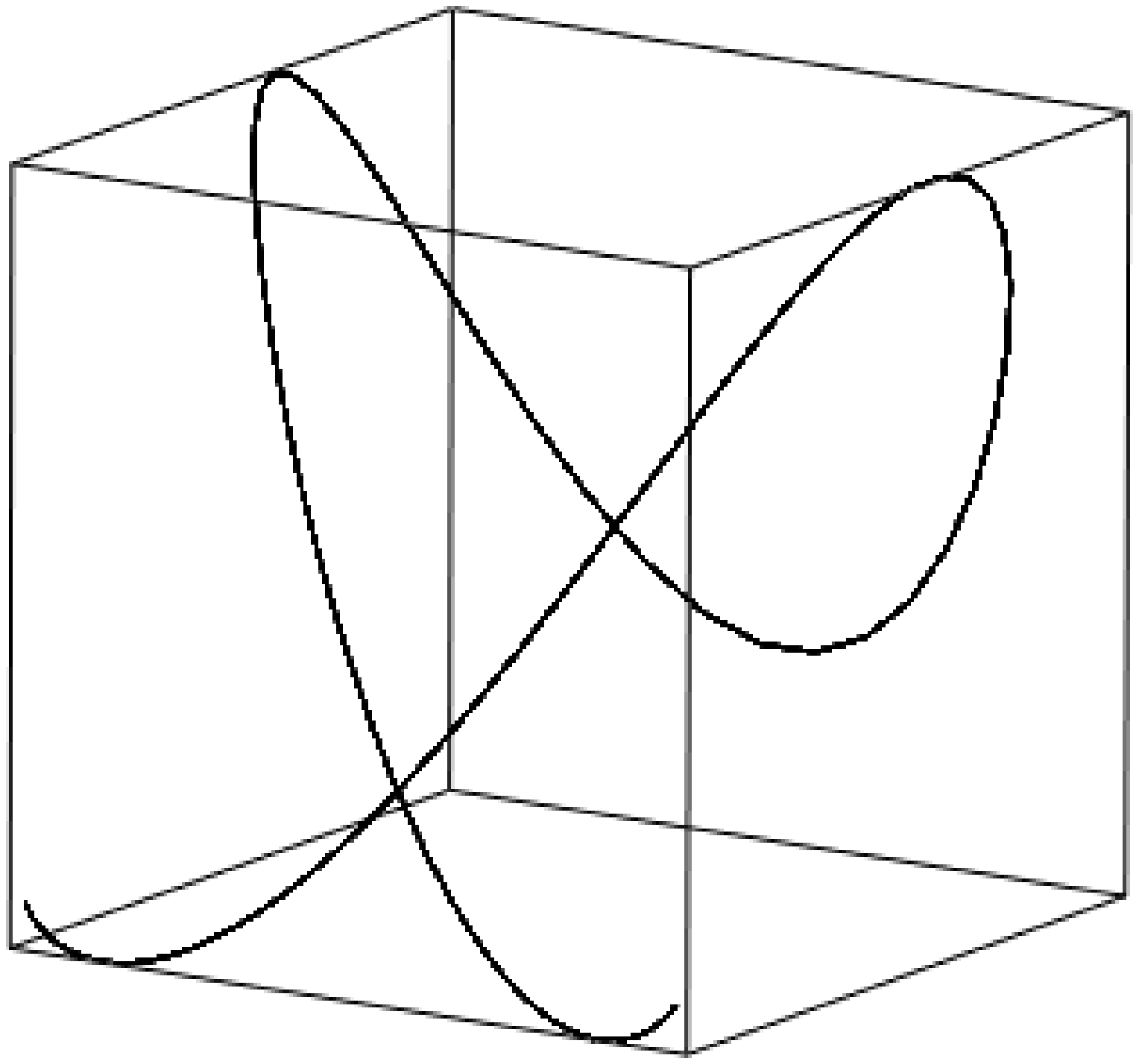}
 \includegraphics[width=0.35\textwidth]{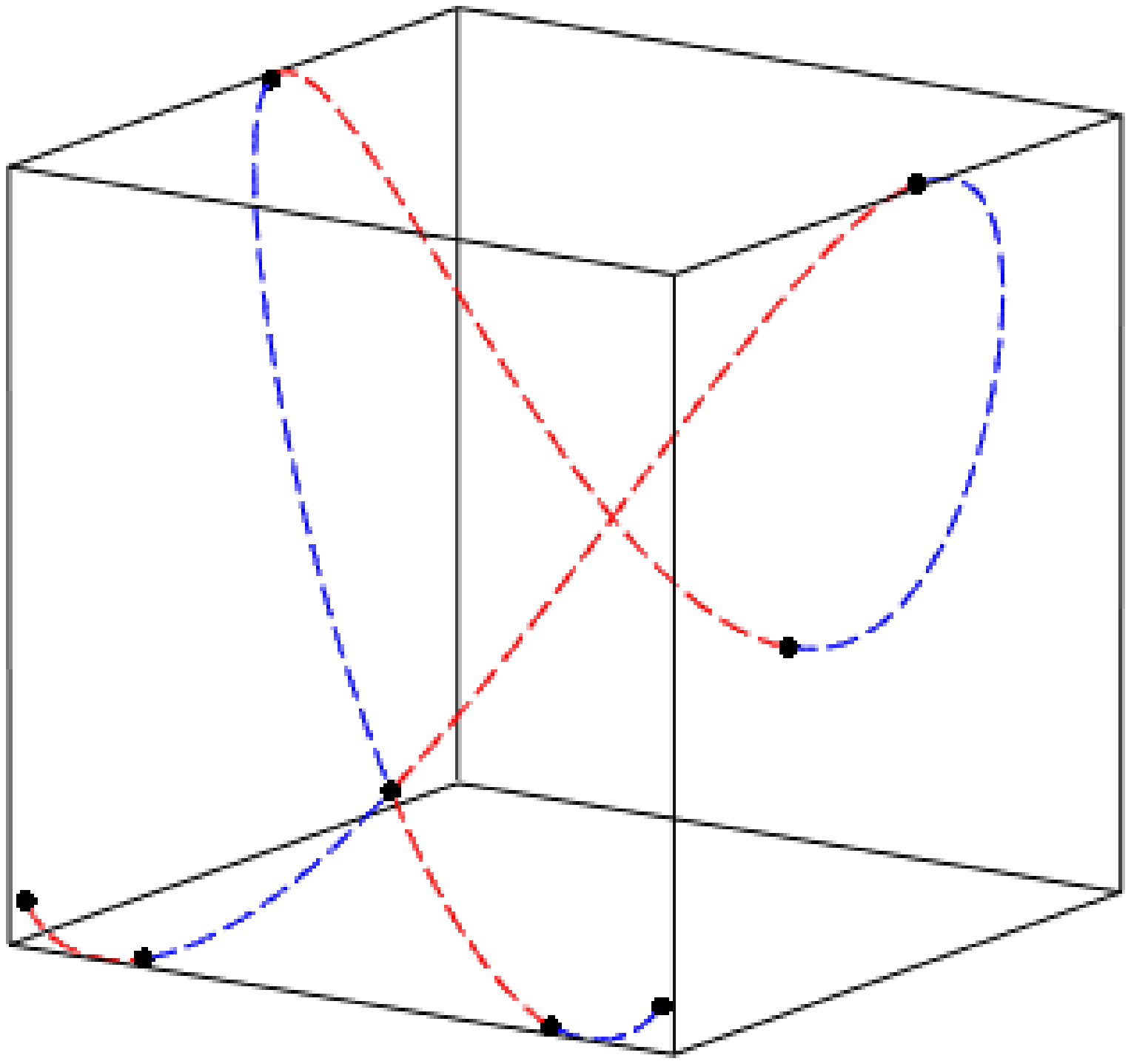}
 \caption{$\r_1(t)$ and $\p(s)$}
 \label{fig3}
\end{figure}

 \begin{figure}[!h]
 \centering
  \includegraphics[width=0.36\textwidth]{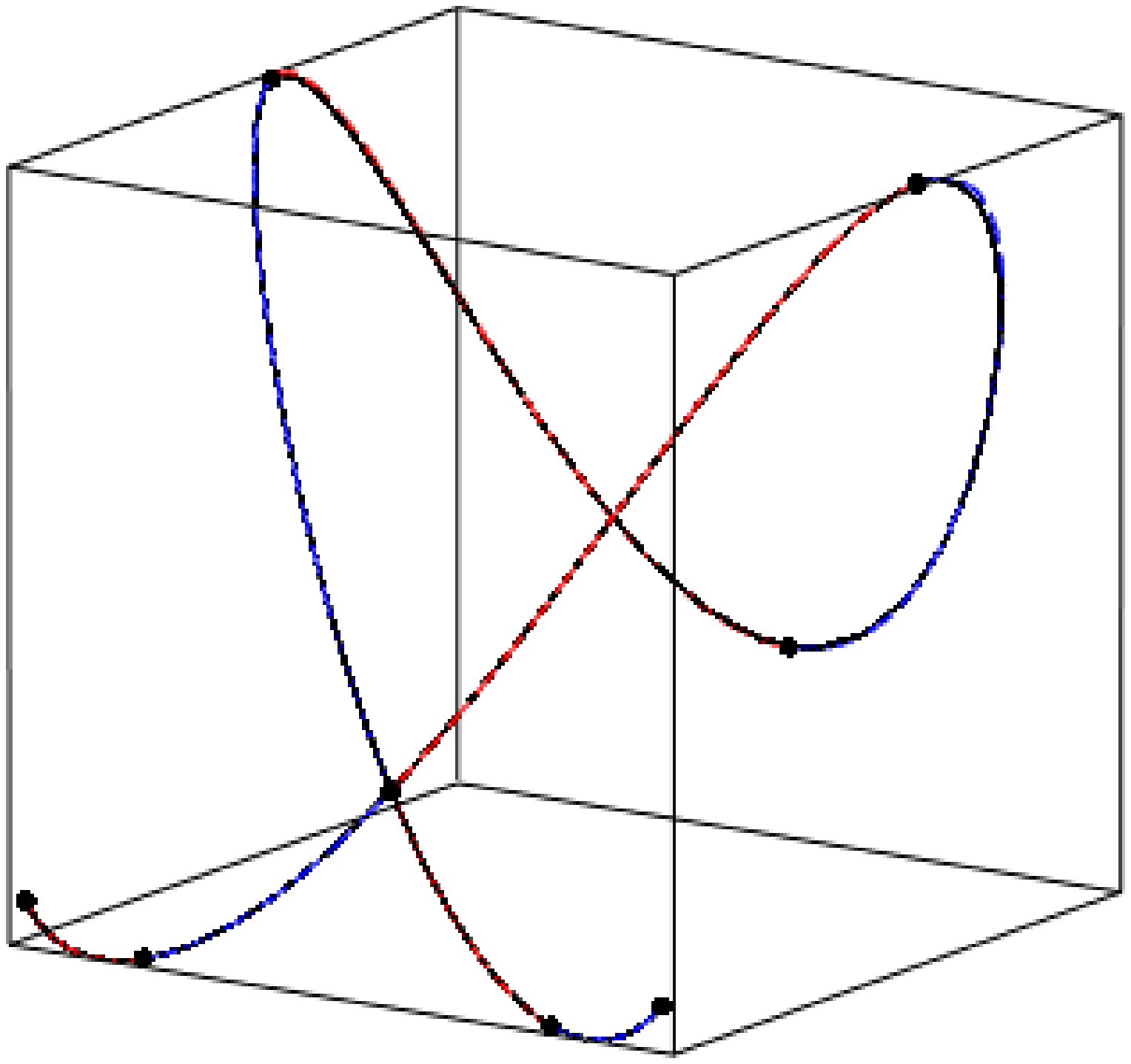}
  \includegraphics[width=0.36\textwidth]{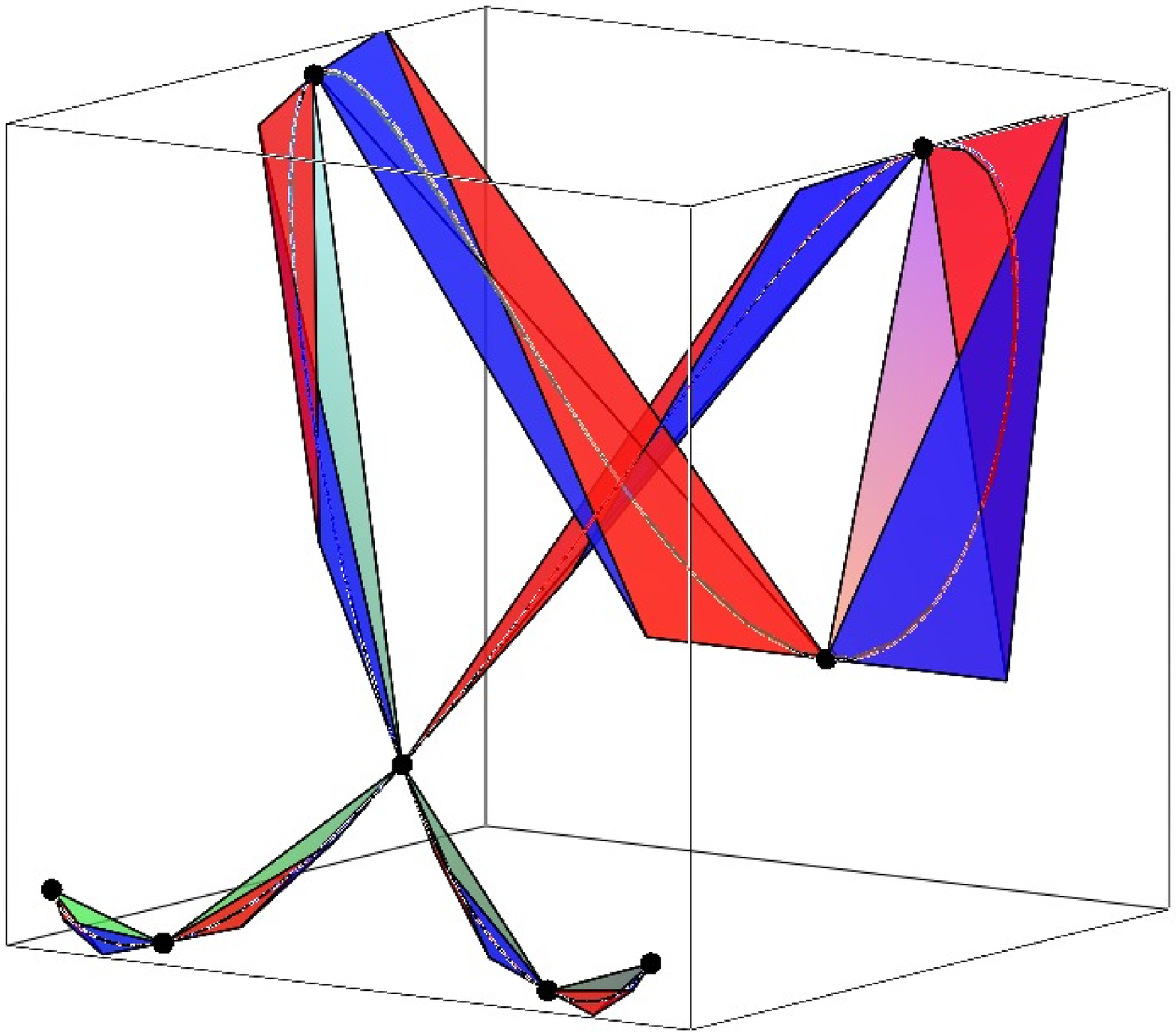}
 \caption{$\r_1(t)$ v.s. $\p(s)$ and control tetrahedron}
 \label{fig31}
\end{figure}

 \begin{figure}[!htp]
 \centering
 \includegraphics[width=0.7\textwidth]{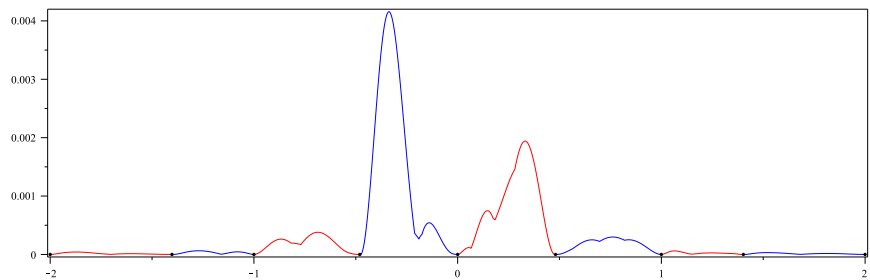}
 \caption{Numerical error for $\r_1$ with $m=300$}
 \label{fig30}
\end{figure}
As we know, the point $(0,0,0)$ is a characteristic point from the topology determining. It is preserved in $\p(s)$ and $\p(s)$ is $C^1$ at this point.
Each corresponding segment of $\p(s)$ and $\r_1(t)$ is interpolated with the Frenet Frames at the endpoints.
One can find that $\r_1(t),t\in [-\infty,+\infty]$ is an asymmetric
space trifolium curve. To approximate the other two parts of
$t\in[-\infty,-2]$ and $t\in[2,+\infty]$, we can transform $t=\pm
\infty$ to $t=0$ by a reparametrization as $t'= 1/t$. Then
approximating $\r_1(t'),t' \in[-1/2,1/2]$ and combining the former
spline segment, we can get the approximation of the whole trifolium
curve.
\end{example}

\begin{example} Two more space curves are given in this example.
$\r_2(t)$ has a complex singular point and $\r_3(t)$ is a random
curve with degree nine.
$$\r_2(t)=\left({\frac {{t}^{2} \left( t-1 \right) ^{2}}{ \left( 1+{t}^{2} \right) ^{
2}}},{\frac {t \left( t-1 \right) ^{3}}{1+{t}^{2}}},{\frac {t \left( t
-1 \right) ^{4}}{1+{t}^{2}}}.
\right), t\in [-1/16,3/2]$$
$$\begin{array}{lcl}\r_3(t)&=&\left( {\frac {t
\left( 1181\,{t}^{8}-1878\,{t}^{7}-1236\,{t}^{6}+1960\,{t}^
{5}+2058\,{t}^{4}-2688\,{t}^{3}+532\,{t}^{2}-9+72\,t \right) }{-2+9\,t
-72\,{t}^{2}+308\,{t}^{3}-840\,{t}^{4}+1218\,{t}^{5}-952\,{t}^{6}+588
\,{t}^{7}-408\,{t}^{8}+149\,{t}^{9}}}, \right.\\
  &&
-{\frac {t \left( -1686\,{t}^{7}
+287\,{t}^{8}+3252\,{t}^{6}-2464\,{t}^{5}+462\,{t}^{4}+168\,{t}^{3}-28
\,{t}^{2}+9 \right) }{-2+9\,t-72\,{t}^{2}+308\,{t}^{3}-840\,{t}^{4}+
1218\,{t}^{5}-952\,{t}^{6}+588\,{t}^{7}-408\,{t}^{8}+149\,{t}^{9}}},\\
 & & \left.
-
\,{\frac {4{t}^{2} \left( 263\,{t}^{7}-924\,{t}^{6}+1338\,{t}^{5}-1190
\,{t}^{4}+861\,{t}^{3}-483\,{t}^{2}+154\,t-18 \right) }{-2+9\,t-72\,{t
}^{2}+308\,{t}^{3}-840\,{t}^{4}+1218\,{t}^{5}-952\,{t}^{6}+588\,{t}^{7
}-408\,{t}^{8}+149\,{t}^{9}}}
\right), t\in [0,1]\end{array}
$$

The approximated curves, approximate spline curves, and the
numerical errors are shown in the following figures
(Figures~\ref{fig4}, \ref{fig40}, \ref{fig41}). In $\r_2(t)$, $(0,0,0)$ is a self-intersected point with $t=0,1$, it is also a cusp point at $t=1$. This point is preserved in our approximate B-spline curve $\p(s)$. Furthermore, the limited tangent directions of the cusp are also preserved. $\p(s)$ is $C^1$ or $C^0$ at $(0,0,0)$ when $\p(s)$ passes through $(0,0,0)$ as a self-intersected or a cusp point respectively.
The approximation
information for curves $\r_1,\r_2$, and $\r_3$ is listed
in~Table~\ref{table}.
 \begin{figure}[!h]
 \centering
 \includegraphics[width=0.35\textwidth]{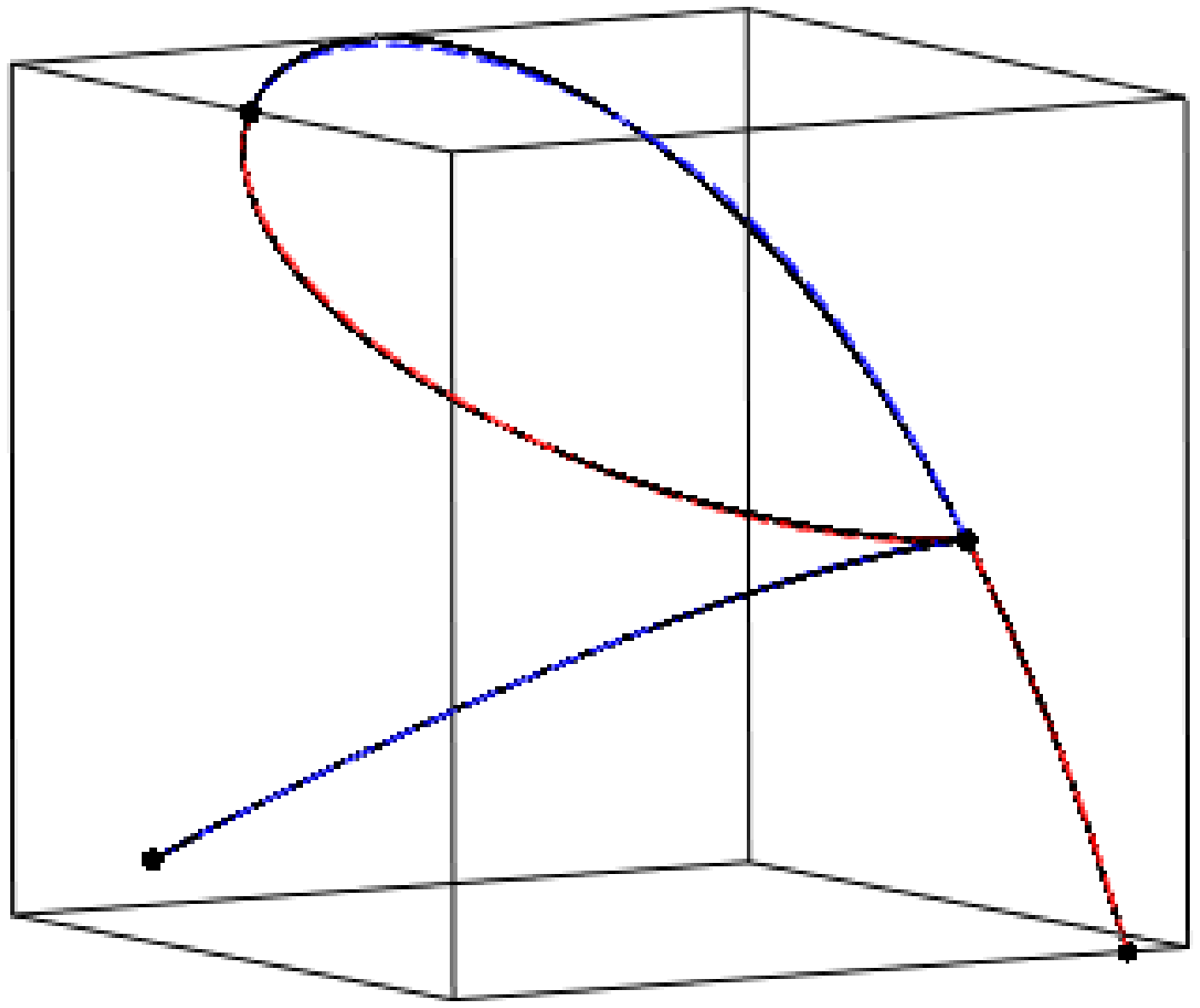}\qquad
 \includegraphics[width=0.34\textwidth]{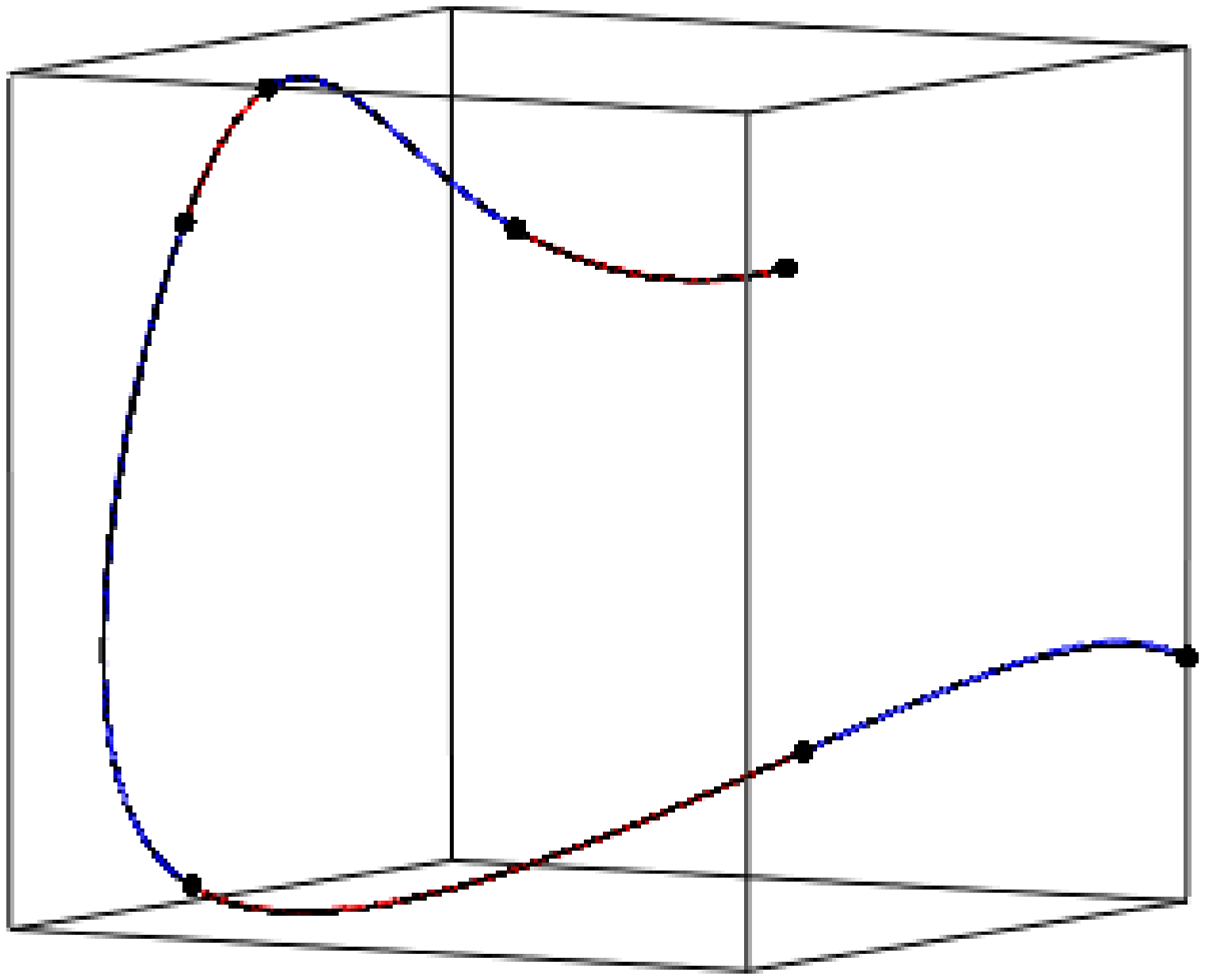}

 \caption{$\r_2(t)$ v.s. $\p_2(s)$ and $\r_3(t)$ v.s. $\p_3(s)$}
 \label{fig4}
\end{figure}

 \begin{figure}[!h]
 \centering
 \includegraphics[width=0.70\textwidth]{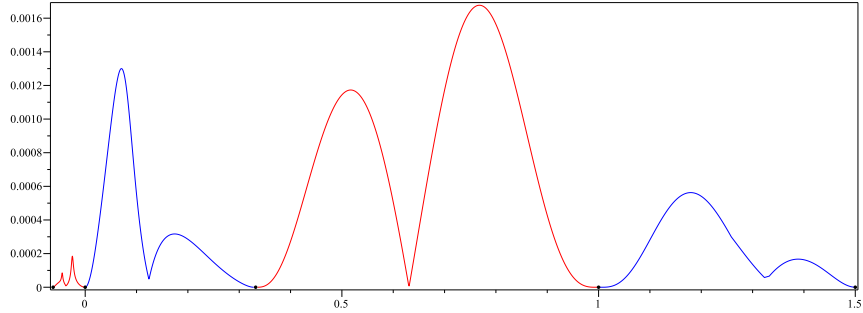}
 \caption{Numerical error for $\r_2$}
 \label{fig40}
\end{figure}

 \begin{figure}[!h]
 \centering
 \includegraphics[width=0.70\textwidth]{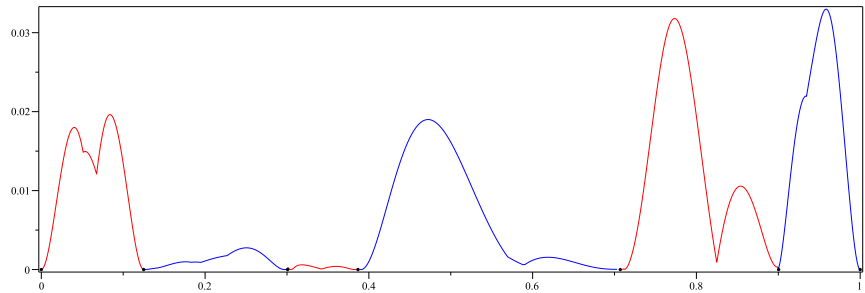}
 \caption{Numerical error for $\r_3$}
 \label{fig41}
\end{figure}

\begin{table}[ht!]
\centering
\begin{tabular}{|c|c|c|c|c|}
\hline curve & degree & error & segments & interval \\
\hline\hline $\r_1$ & $8$ & $0.004157$ & $8$ & $[-2,2]$ \\[0.1cm]
\hline $\r_2$ & $5$ & $0.0001677$ & $4$ &
$[-\frac{1}{16},\frac{3}{2}]$ \\[0.1cm]
\hline $\r_3$ & $9$ & $0.03298$ & $6$ & $[0,1]$ \\
\hline
\end{tabular}\caption{Numerical Approximation}
\label{table}
\end{table}

\end{example}
\section{Conclusion and further work}
We present an algorithm to construct a rational cubic B-spline
approximation for a space parametric curve.
The main purpose of the work is to present an isotopic approximation method
which preserves the geometric features of the original curve.
The approximated curve is divided into quasi-cubic segments which have
similar properties to those of a cubic B\'{e}zier curve. Sufficient
conditions are proposed for a divided segment having the expected
properties and then its approximate B\'{e}zier spline is naturally
constructed. Based on these properties, the shoulder point
approximate algorithm is presented and it is proved to be
convergent. An approximate implicitization can be found by the
$\mu$-basis method. The method is applicable for any parametric
space curve in theory, although the given conditions are more
difficult to compute when the parametric expression is not in
rational form.

The intersection curve of a parametric surface and an implicit
surface is another important type of space curves. The curve can be
regarded as parametric form with two parameters and a constraint
function for them. As a further work,  we will study the
approximation of this type of space curve.

\section*{Acknowledgements}
This work is partially supported by National Natural Science Foundation of
China under Grant 10901163, 11101411, 60821002,
 a National Key Basic Research Project of China (2011CB302400) and a China Postdoctoral Science Foundation.
The authors also wish to
thank the anonymous reviewers for their helpful comments and
suggestions.


\end{document}